\documentclass[letterpaper,twocolumn,10pt]{article}
\usepackage{usenix-2020-09}

\usepackage{amsmath,amsfonts, amstext, amsthm, tikz}
\usepackage{multirow, adjustbox}
\usepackage{xspace, url, comment}
\usepackage{subcaption}
\usepackage[inline]{enumitem}
\usepackage{hyperref}
\usepackage{algorithm}
\usepackage[noend]{algpseudocode}
\usepackage{pgfplots}
\usepackage{booktabs}
\usepackage{bm}
\usepackage[normalem]{ulem} 
\usepackage{thmtools, thm-restate}
\usepackage{ifthen}
\usepackage{ifthen}
\newboolean{extended}

\setboolean{extended}{true}

\ifthenelse{\boolean{extended}}{}{
\usepackage[available]{usenixbadges}
\pagenumbering{gobble}} 
\pgfplotsset{compat=1.17}
\usepackage [%
n,
advantage,
operators,
sets,
adversary,
landau,
probability,
notions,
logic,
ff,
mm,
primitives,
events,
complexity,
asymptotics,
keys]{cryptocode}
\createprocedureblock{game}{center,boxed}{}{}{}
\usepackage{varwidth}
\begin{document}

\newcommand{\TODO}[1]{{\color{red}\texttt{TODO}: #1}}
\newcommand{\diaa}[2][]{\ifthenelse{\equal{#1}{}}{}{ \color{red} \sout{#1}} \color{blue} #2 \color{black}}

\newcommand{\DP}[1]{\ensuremath{#1\mbox{-}\mathsf{DP}}\xspace}
\newcommand{\eDP}{\DP{\epsilon}}
\newcommand{\edDP}{\DP{(\epsilon,\delta)}}
\newcommand{\difp}{differential privacy\xspace}
\newcommand{\km}{$k$-means\xspace}
\newcommand{\dpl}{DPLloyd\xspace}
\newcommand{\SMC}{\text{SMPC}}

\newcommand{\E}{\ensuremath{\mathsf{E}}\xspace}
\newcommand{\EE}[1]{\ensuremath{\mathsf{E}\left[#1\right]}\xspace}
\newcommand{\mse}[1]{\ensuremath{\mathsf{MSE}\left(#1\right)}\xspace}
\newcommand{\Var}[1]{\ensuremath{\mathsf{Var}\left(#1\right)}\xspace}
\newcommand{\Lap}[1]{\ensuremath{\mathsf{Lap}\left(#1\right)}\xspace}

\newcommand\NumPartys{\ensuremath{M}}
\newcommand\PartyIndex{\ensuremath{i}}
\newcommand\PartySet{\ensuremath{p}}
\newcommand{\Party}[1]{\ensuremath{\PartySet_{#1}}}

\newcommand\Server{S}
\newcommand\shares[1]{[[#1]]}  

\newcommand\DatasetSize{\ensuremath{N}}
\newcommand\DatasetIndex{\ensuremath{l}}
\newcommand\Dataset{\ensuremath{D}}
\newcommand\datapoint{\ensuremath{x}}
\newcommand\datapointi[1]{\ensuremath{x_{#1}}}

\newcommand{\VV}[1]{\ensuremath{v_{#1}}}
\newcommand{\round}[1]{\ensuremath{\lceil{#1}\rfloor}}
\newcommand{\DPNoise}{\ensuremath{\gamma}}

\newcommand{\fixedVV}[1]{\ensuremath{\tilde{\VV{#1}}}}
\newcommand\Prec{\ensuremath{q}}
\newcommand\SF{\ensuremath{2^{\Prec}}}

\newcommand\silscore{\ensuremath{\mathbf{Sil}}}

\newcommand\Ddim{\ensuremath{d}}
\newcommand\DimI{\ensuremath{h}}

\newcommand\bound{\ensuremath{B}}

\newcommand\NumClusters{\ensuremath{k}}
\newcommand\ClusterID{\ensuremath{j}}
\newcommand\Clusters{\ensuremath{\mathbf{O}}}
\newcommand{\ClusterI}[1]{\ensuremath{O_{#1}}}

\newcommand\SmS{\ensuremath{S}}
\newcommand\CntS{\ensuremath{C}}
\newcommand\DiffS{\ensuremath{R}}

\newcommand{\Sm}[1]{\ensuremath{\SmS_{#1}}}
\newcommand{\Cnt}[1]{\ensuremath{\CntS_{#1}}}
\newcommand{\Diff}[1]{\ensuremath{\DiffS_{#1}}}
\newcommand{\NoiseSm}[1]{\ensuremath{\tilde{\SmS}_{#1}}}
\newcommand{\NoiseCnt}[1]{\ensuremath{\tilde{\CntS}_{#1}}}
\newcommand{\NoiseDiff}[1]{\ensuremath{\tilde{\DiffS}_{#1}}}

\newcommand{\EX}[1]{\ensuremath{\mathbb{E}\left[#1\right]}}
\newcommand\Centers{\ensuremath{\bm{\mu}}}
\newcommand\NoiseCenters{\ensuremath{\tilde{\bm{\mu}}}}
\newcommand{\CenterI}[1]{\ensuremath{\mu_{#1}}}
\newcommand{\NoiseCenterI}[1]{\ensuremath{\tilde{\mu}_{#1}}}

\newcommand\NumIter{\ensuremath{T}}
\newcommand\IT{\ensuremath{t}}
\newcommand\MXD{\ensuremath{\eta}}

\newcommand\CF{\ensuremath{J}}

\newcommand\MS{\ensuremath{r}}

\newcommand{\MM}[1]{\ensuremath{\MS_{#1}}}
\newcommand\Ddiag{\ensuremath{\beta}}
\newcommand{\ourProtocol}{FastLloyd}
\newcommand{\suProtocol}{SuLloyd}
\newcommand{\gProtocol}{GLloyd}
\newcommand{\mohProtocol}{MohLloyd}
\newcommand{\lloydsProtocol}{Lloyd}


\newtheorem{theorem}{Theorem}[section]
\newtheorem{lemma}[theorem]{Lemma}
\newtheorem{corollary}[theorem]{Corollary}
\newtheorem{remark}{Remark}[section]
\newtheorem{example}{Example}[section]
\newtheorem{definition}{Definition}[section]


\newcommand{\generateMetricFigureGroup}[3][true]{%
    \begin{figure*}[t]
        \centering
        \ifthenelse{\equal{#1}{true}}{%
            \begin{subfigure}{\textwidth}
                \centering
                \includegraphics[width=0.5\textwidth]{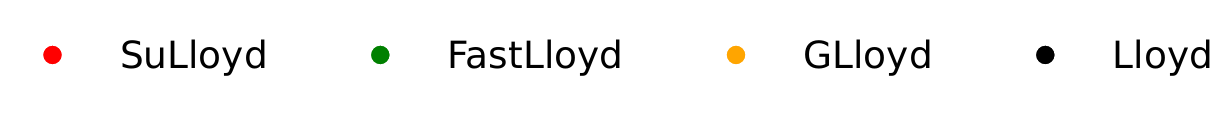}
            \end{subfigure}
        }{}%
        \begin{subfigure}{0.19\textwidth}
            \includegraphics[width=\textwidth]{figs/metrics/#2/adult.pdf}
            \caption{Adult}
            \label{fig:#2_adult}
        \end{subfigure}
        \begin{subfigure}{0.19\textwidth}
            \includegraphics[width=\textwidth]{figs/metrics/#2/birch2.pdf}
            \caption{Birch2}
            \label{fig:#2_birch2}
        \end{subfigure}
        \begin{subfigure}{0.19\textwidth}
            \includegraphics[width=\textwidth]{figs/metrics/#2/iris.pdf}
            \caption{Iris}
            \label{fig:#2_iris}
        \end{subfigure}
        \begin{subfigure}{0.19\textwidth}
            \includegraphics[width=\textwidth]{figs/metrics/#2/breast.pdf}
            \caption{Breast}
            \label{fig:#2_breast}
        \end{subfigure}
        \begin{subfigure}{0.19\textwidth}
            \includegraphics[width=\textwidth]{figs/metrics/#2/yeast.pdf}
            \caption{Yeast}
            \label{fig:#2_yeast}
        \end{subfigure}
        \begin{subfigure}{0.19\textwidth}
            \includegraphics[width=\textwidth]{figs/metrics/#2/house.pdf}
            \caption{House}
            \label{fig:#2_house}
        \end{subfigure}
        \begin{subfigure}{0.19\textwidth}
            \includegraphics[width=\textwidth]{figs/metrics/#2/lsun.pdf}
            \caption{LSun}
            \label{fig:#2_lsun}
        \end{subfigure}
        \begin{subfigure}{0.19\textwidth}
            \includegraphics[width=\textwidth]{figs/metrics/#2/s1.pdf}
            \caption{S1}
            \label{fig:#2_s1}
        \end{subfigure}
        \begin{subfigure}{0.19\textwidth}
            \includegraphics[width=\textwidth]{figs/metrics/#2/wine.pdf}
            \caption{Wine}
            \label{fig:#2_wine}
        \end{subfigure}
        \begin{subfigure}{0.19\textwidth}
            \includegraphics[width=\textwidth]{figs/metrics/#2/mnist.pdf}
            \caption{MNIST}
            \label{fig:#2_mnist}
        \end{subfigure}
        \caption{#3 vs $\epsilon$ for real datasets.}
        \label{fig:#2_util_compare}
    \end{figure*}
}


\date{}

\title{\ourProtocol: Federated, Accurate, Secure, and Tunable $k$-Means Clustering with Differential Privacy}

\author{
{\rm Abdulrahman Diaa}\\
University of Waterloo\\
abdulrahman.diaa@uwaterloo.ca
\and
{\rm Thomas Humphries}\\
University of Waterloo\\
thomas.humphries@uwaterloo.ca
\and
{\rm Florian Kerschbaum}\\
University of Waterloo\\
florian.kerschbaum@uwaterloo.ca}

\maketitle

\begin{abstract}
We study the problem of privacy-preserving $k$-means clustering in the horizontally federated setting. Existing federated approaches using secure computation suffer from substantial overheads and do not offer output privacy. At the same time, differentially private (DP) $k$-means algorithms either assume a trusted central curator or significantly degrade utility by adding noise in the local DP model. Naively combining the secure and central DP solutions results in a protocol with impractical overhead. Instead, our work provides enhancements to both the DP and secure computation components, resulting in a design that is faster, more private, and more accurate than previous work. By utilizing the computational DP model, we design a lightweight, secure aggregation-based approach that achieves five orders of magnitude speed-up over state-of-the-art related work. Furthermore, we not only maintain the utility of the state-of-the-art in the central model of DP, but we improve the utility further by designing a new DP clustering mechanism.
\end{abstract}

\section{Introduction}
Unsupervised learning allows data analysts to extract meaningful patterns from raw data that may be difficult to label.
The canonical example of unsupervised learning is the Euclidean \km clustering problem, where data is grouped into clusters with similar features.
Clustering has a plethora of important applications in recommendation systems, fraud detection, and healthcare analytics~\cite{ghosal2020short}.
In such applications, the dataset often contains sensitive data, which necessitates the use of privacy-preserving techniques.
Combining databases horizontally split across multiple parties in the federated scenario yields more robust insights about the global population at the cost of 
further exasperating the privacy risks. 

To counteract the privacy risks, a number of works focus on solving the federated \km problem using secure multi-party computation~\cite{Vaidya2003, Jha2005, Jagannathan2005, Jagannathan2010, Gheid2016, Bunn2007, Jschke2018, Mohassel2020, Patel2012, Yuan2019, Rao2015, Liu2015, Silva2017, Jiang2020}.
Secure computation enables analysts to solve the \km problem while keeping the sensitive data encrypted.
However, all of these works still output the exact result of the clustering.
This gives a false sense of privacy; although the input and intermediate computations all remain private, no effort is made to protect the privacy of the output.
It is well known in the literature that publishing exact statistics can leak significant information about the input data and allow attacks such as dataset reconstruction~\cite{Dwork2017}.
This level of privacy does not justify the significant runtime overhead incurred by solely using secure computation.

A separate line of research considers using differential privacy (DP) to protect the
data.
Differential privacy guarantees that the output (the cluster centroids) will be approximately the same regardless of any individual user's participation. 
Differential privacy can be applied in the central, local, or shuffle models.
In the central model, a trusted aggregator receives the input as plaintext (unprotected) and randomizes the cluster centroids to ensure the output satisfies DP. 
Several works~\cite{Blum2005, Dwork2011, Su2016, Mohan2012, Zhangj2013, Balcan2017, Lu2020, Park2017, NI2021} have considered the central setting.
In the local model, instead of using a trusted aggregator, each party perturbs their data locally before sending it to the aggregator.
While this removes the need for a trusted aggregator (similar to secure computation), the utility at high privacy levels is low.
Specifically, the utility is asymptotically worse than the central model by a factor of $\sqrt{n}$, where $n$ is the number of data points~\cite{chanOptimalLowerBound2012}.
Several works~\cite{Nissim2018, Stemmer2021, Chang2021} have considered the local model.
The shuffle model is a hybrid between the local and central models, with two mutually distrustful parties: a shuffler and an aggregator.
Despite improving over the local model, the shuffle model often has a worse utility than the central model and requires additional parties or computational overhead to implement an oblivious shuffle.

To summarize, no prior work offers a solution that protects both the input and output privacy of the federated \km problem and provides a good utility vs.~privacy trade-off.
A straw man solution to our problem would be to combine the current state-of-the-art solutions in federated \km using secure computation and \km in the central DP model.
However, this solution would be unacceptable as the runtime of current federated \km approaches using secure computation is prohibitively large.
The state-of-the-art approach of Mohassel et al.~\cite{Mohassel2020} incurs runtime overheads on the order of tens of minutes.
Naively adding differentially private noise to this computation will further increase this runtime.
%

Our work focuses on designing an efficient DP \km in the federated setting (using secure computation) that protects the privacy of the input, output, and intermediate computations.
Following related work in federated \km~\cite{Mohassel2020, Jiang2020}, we use Lloyd's algorithm as the foundation of our protocol.
However, we design a new variant of DP-Lloyd~\cite{Blum2005, Su2016} that achieves a tighter bound on the sensitivity.
Our algorithm significantly improves the clustering utility over the state-of-the-art DP algorithm of Su et al.~\cite{Su2016}, especially in higher dimensions and number of clusters.

Our algorithm is of independent interest in the central model of DP.
First, we enforce a bound on the radius of each cluster in the assignment step of Lloyd's algorithm.
We then modify Lloyd's algorithm to compute updates relative to the previous centroid.
Our relative updates have a sensitivity proportional to the bounded cluster radius rather than the domain size used in previous work~\cite{Blum2005, Su2016}.
We find that our radius-based sensitivity bound is best suited to the use of the analytic Gaussian mechanism~\cite{balle18_gaussian,dong22_gdp}, which further improves utility.
Finally, we utilize additional post-processing steps based on the radius and the domain to improve the algorithm further.
Together, these components lead to an improvement of up to 88\% in utility (reduction in clustering error). 

To tailor our DP-Lloyd algorithm to secure computation, we design a protocol that publishes intermediate computations in each iteration.
This allows operations that would be expensive in secure computation (such as assignment and division) to be conducted in plain text.
However, unlike prior work~\cite{Gheid2016, Jiang2020}, we ensure that the intermediate computations are protected by DP guarantees.
We prove that the protocol is secure in the computational DP definition~\cite{mironov_CDP_09}, which allows DP-bounded leakage during the computation.
Typically, utilizing the computational DP model to leak intermediate steps implies sacrificing privacy or utility for an efficient runtime. 
However, in our case, we specifically choose to leak intermediate values already accounted for in the central DP proof of DP-Lloyd~\cite{Blum2005, Su2016}.
In other words, we get this speed-up with no additional cost to the privacy parameter $\epsilon$.
Furthermore, rather than decrease the utility, we improve the utility of the clustering over state-of-the-art approaches~\cite{Su2016}.

To perform the aggregation over multiple parties, we design a lightweight secure aggregation protocol that keeps the aggregator oblivious to the client's inputs and the global centroids that are output at each iteration of the protocol.
This allows the aggregator to add DP noise equivalent to the central model of DP.
The combination of our improved DP algorithm, leaking the DP centroids, and our lightweight, secure aggregation protocol allows us to reduce the computation time by up to five orders of magnitude compared to related work~\cite{Mohassel2020} (from minutes to milliseconds per iteration).

In summary, our contributions are four-fold:

\begin{itemize}
    \item We design the first DP protocol for horizontally federated private \km.
    \item We improve the clustering utility over state-of-the-art DP \km solutions by developing a new DP algorithm with various improvements, such as enforcing a radius constraint on the centroids and using relative cluster updates.
    \item We design an efficient protocol using DP and a lightweight secure aggregation protocol to implement our protocol in the local trust model. 
    \item We prove our protocol is secure and preserves the end-to-end privacy in the computational model of DP and reduces the runtime by five orders of magnitude over the state-of-the-art secure federated approaches.
\end{itemize}

The remainder of the paper is organized as follows.
We begin with some problem-specific background information in Section~\ref{sec:background}
 and formally define the problem in Section~\ref{sec:prob_state}. In Section~\ref{sec:rel_work}, we summarize the relevant literature.
We then present our complete protocol for federated DP \km (\ourProtocol) and prove its privacy and utility in Section~\ref{sec:main_algo}. 
Finally, in Section~\ref{sec:eval}, we give an in-depth evaluation of our protocol in terms of utility, runtime, and communication size over various real-world and synthetic datasets.

\section{Background}\label{sec:background}
\subsection{Notation}
Throughout this paper, we will denote objects that we intend to further slice/index with boldface notation (e.g. $\Centers$), while keeping atomic objects (whether scalars or vectors) as normal face (e.g. $\epsilon$).
A summary of the notation used in this paper is provided in Table~\ref{tab:notation}, and we will define the notation as it is used. 
\begin{table}
    \centering
    \begin{tabular}{|c|c|}
    \hline
    \textbf{Symbol} &\textbf{Description} \\ \hline
    $\boldsymbol{\PartySet}$ &  Set of $\NumPartys$ clients: $\{\Party{1}, \Party{2}, ..., \Party{\NumPartys}\}$ \\ \hline  
    \Server& Service provider \\ \hline
    $\Dataset_{\PartyIndex}$ & Local dataset of client $\Party{\PartyIndex}$: $\{\datapointi{1}, \datapointi{2}, ..., \datapointi{\DatasetSize_{\PartyIndex}}\}$ \\ \hline
    $\Ddim$ & Dimension of the data points \\ \hline
    $\NumClusters$ & Number of clusters \\ \hline
    $\Clusters$ & Set of \NumClusters{}~clusters: $\{\ClusterI{1}, \ClusterI{2}, ..., \ClusterI{\NumClusters}\}$ \\ \hline
    $\Centers$ & Centroids of clusters: $\{\CenterI{1}, \CenterI{2}, ..., \CenterI{\NumClusters}\}$ \\ \hline
    $\Sm{\ClusterID}$ & Sum of data points in cluster $\ClusterID$ \\ \hline
    $\Cnt{\ClusterID}$ & Count of data points in cluster $\ClusterID$ \\ \hline
    $\Diff{\ClusterID}$ & Relative sum of data points in cluster $\ClusterID$ \\ \hline
    $\NumIter$ & Number of iterations in Lloyd's algorithm \\ \hline
    $\bm{\DPNoise}$& Differentially-private noise \\ \hline
    $\MXD$ & Maximum radius bound \\ \hline
    $(\epsilon, \delta)$& Privacy budget \\ \hline
    $\sigma$& Noise multiplier \\ \hline

    \end{tabular}
    \caption{Notation used in the paper.}\label{tab:notation}
\end{table}

\subsection{$k$-Means Problem}\label{sec:lloyd}
The \km problem is a discrete optimization problem that aims to partition a set of $\DatasetSize$ $\Ddim$-dimensional observations into $\NumClusters$ clusters, each represented by its mean or centroid.
Given a dataset of observations $\Dataset = \{\datapointi{1}, \datapointi{2}, ..., \datapointi{\DatasetSize}\}$, where each observation is a $\Ddim$-dimensional real vector, the \km problem is to find an assignment of data points to clusters, and a set of cluster centroids, that minimizes the Within-cluster Sum of Squares (WCSS) objective:

\begin{equation}\label{eqn:wcss_obj}
  \underset{\Clusters, \Centers}{\text{argmin}} \sum_{\ClusterID=1}^{\NumClusters} \sum_{\datapointi{\DatasetIndex} \in \ClusterI{\ClusterID}} ||\datapointi{\DatasetIndex} - \CenterI{\ClusterID}||^2_2 
\end{equation}

where $\Clusters = \{\ClusterI{1}, \ClusterI{2}, ..., \ClusterI{\NumClusters}\}$, $\ClusterI{\ClusterID} \subseteq \Dataset$ are the clusters, and  $\Centers = \{\CenterI{1}, \CenterI{2}, ..., \CenterI{\NumClusters}\}$ are the centroids of the clusters, defined to be the (arithmetic) mean of points in $\ClusterI{\ClusterID}$.

\subsubsection{Lloyd's Algorithm}
The \km problem is NP-hard in Euclidean space~\cite{Aloise2009}. However, a simple heuristic known as Lloyd's algorithm is commonly applied for practical applications~\cite{Lloyd1982}.
We detail the basic procedure of Lloyd's algorithm:

\begin{enumerate}
  \item Initialization step: Randomly sample $\NumClusters$ centroids $\{\CenterI{1}, \CenterI{2}, ..., \CenterI{\NumClusters}\}$ (typically from the datapoints).
  \item Repeat until convergence (the assignments no longer change) or a predetermined number of iterations has been reached:
        \begin{enumerate}
          \item Assignment step: Assign each observation $\datapointi{\DatasetIndex}$ to the nearest centroid (using the Euclidean distance). This creates clusters $\ClusterI{\ClusterID}$ for $\ClusterID=1,2,...,\NumClusters$.
          Formally, the assignment is:
                \begin{equation*}
                    \displayindent0pt
                    \displaywidth\columnwidth
                  \ClusterI{\ClusterID}^{(\IT)} = \big \{ \datapointi{\DatasetIndex} :
                  ||\datapointi{\DatasetIndex} - \CenterI{\ClusterID}^{(\IT-1)}||^2_2 \leq ||\datapointi{\DatasetIndex}- \CenterI{c}^{(\IT-1)}||^2_2 \ \forall c\in\{1,\dots, \NumClusters \}\}
                \end{equation*}
          \item Update step: Calculate the new centroids to be the mean of the observations in the cluster:
                \begin{equation*}
                    \displayindent0pt
                    \displaywidth\columnwidth
                  \CenterI{\ClusterID}^{(\IT)} = \frac{\sum\limits_{\datapointi{\DatasetIndex} \in \ClusterI{\ClusterID}^{(\IT)}} \datapointi{\DatasetIndex}}{\left|\ClusterI{\ClusterID}^{(\IT)}\right|} 
                \end{equation*}
        \end{enumerate}
\end{enumerate}

\subsection{Differential Privacy}
Differential privacy (DP)~\cite{dwork2006} is an increasingly popular notion to protect the privacy of individuals while allowing the computation of aggregate statistics.
Differential privacy guarantees that an algorithm's output is approximately the same, regardless of the participation of any single user.
More formally, differential privacy can be defined as follows.
\begin{definition}[Differential Privacy]\label{def.adp}
    A randomized algorithm $M:\mathcal{D} \mapsto \mathbb{R}$ is $(\epsilon,\delta)$-DP, if for any pair of neighbouring datasets $\Dataset,\Dataset' \in \mathcal{D}$, and for any $S \subseteq \mathbb{R}$ we have
    \begin{equation}
     \Pr[M(\Dataset)\in S] \leq e^{\epsilon} \Pr[M(\Dataset')\in S] +\delta .
    \end{equation}
\end{definition}
The privacy parameter $\epsilon$ defines how similar the outputs must be, and $\delta$ allows a small chance of failure in the definition.
We use the unbounded neighbouring definition where datasets are neighbours if $|\Dataset\backslash \Dataset'\cup \Dataset'\backslash \Dataset| = 1$. That is, we allow for the addition or removal of a single data point. 
We note that arbitrary computations can be carried out on the output of a DP mechanism without affecting privacy (the post-processing lemma~\cite{dwork_roth_DP_Txtbook}).
Finally, DP is composed naturally with multiple runs of a mechanism.
If we apply a differentially private mechanism(s) sequentially, the privacy parameters are composed through summation or more advanced methods~\cite{dwork_roth_DP_Txtbook}.
If a mechanism is applied multiple times over disjoint subsets of the dataset, then the total privacy leakage is the maximum privacy parameter over each subset (parallel composition~\cite{dwork_roth_DP_Txtbook}).

\begin{definition}[Sensitivity]
    Let $f:\mathcal{D}\mapsto \mathbb{R}^k$. If $\ \mathbb{D}$ is a distance metric between elements of $\ \mathbb{R}^k$ then the $\mathbb{D}$-sensitivity of $f$ is
    \begin{equation}
        \Delta^{(f)}=\max_{(D,D')} \mathbb{D}(f(D), f(D')),
    \end{equation}
    where $(D,D')$ are pairs of neighbouring datasets.
\end{definition}
We will focus on the $\ell_2$ norm in this work as we use the Gaussian Mechanism. To analyze the Gaussian Mechanism, we will use Gaussian Differential Privacy (GDP)~\cite{dong22_gdp}.

\begin{definition}[GDP~\cite{dong22_gdp}] \label{def:GDP}
    A mechanism $M$ is said to satisfy $\theta$-Gaussian Differential Privacy ($\theta$-GDP) if it is $G_\theta$-DP. That is,
    \[
    \mathcal{T}\big(M(D),M(D')\big) \ge G_\theta
    \]
    for all neighbouring datasets $D$ and $D'$, where $\mathcal{T}$ is a trade-off function measuring the difficulty for attackers in identifying presence of an individual data point and $G_\theta=\mathcal{T}\big(\mathcal{N}(0,1),\mathcal{N}(\theta,1)\big)$ (see Dong et al.~\cite{dong22_gdp} for specifics of the definition).
\end{definition}
 Naturally, the Gaussian Mechanism satisfies GDP.
\begin{theorem}(Gaussian Mechanism GDP~\cite{dong22_gdp})\label{thm:g_mech}
    Define the Gaussian mechanism that operates on a statistic $f$ as $M(D) = f(D) + \DPNoise$, where $\DPNoise \sim \mathcal{N}(0, (\Delta^{(f)})^2/\theta^2)$. Then, $M$ is $\theta$-GDP.
\end{theorem}
Similar to DP, GDP composes over multiple adaptive uses of a mechanism.
\begin{theorem}[GDP Composition~\cite{dong22_gdp}]\label{thm:gdp_comp}
    The $n$-fold composition of $\theta_i$-GDP mechanisms is $\sqrt{\theta_1^2+\cdots+\theta_n^2}$-GDP.
\end{theorem}
Finally, it is possible to convert a GDP guarantee to DP and vice versa:
\begin{theorem}[GDP to DP~\cite{dong22_gdp, balle18_gaussian}]\label{thm:GDPtoDP}
    A mechanism is $\theta$-GDP if and only if it is $\big(\epsilon,\delta(\epsilon)\big)$-DP for all $\epsilon \geq 0$, where
    \[
    \delta(\epsilon)= \Phi\Big( -\frac{\varepsilon}{\theta} +\frac{\theta}{2} \Big)-
    e^{\epsilon}\Phi\Big(- \frac{\varepsilon}{\theta} - \frac{\theta}{2} \Big).
    \]
\end{theorem}
In practice, we use the algorithm derived by Balle and Wang to solve this function for $\theta$~\cite[Algorithm 1]{balle18_gaussian}.

\section{Problem Statement}\label{sec:prob_state}
We consider \NumPartys~clients denoted by $\boldsymbol{\PartySet} = \{\Party{1}, \Party{2}, ..., \Party{\NumPartys}\}$, with each party \Party{\PartyIndex} owning a private dataset $\Dataset_{\PartyIndex} = \{\datapointi{1}, \datapointi{2}, ..., \datapointi{\DatasetSize_{\PartyIndex}}\}$, where $\datapointi{\DatasetIndex} \in \mathbb{R}^d$, $\Ddim$ denotes the dimensionality of the dataset and $\DatasetSize_{\PartyIndex}$ is the size of the dataset of party \PartyIndex{}. 
The total dataset that we compute on, can be very large, even if there are only a small number of clients.
The objective is to compute a collaborative \km clustering $(\Clusters, \Centers)$ over $\boldsymbol{\Dataset}$ that minimizes the WCSS objective (\ref{eqn:wcss_obj}); where $\boldsymbol{\Dataset}= \bigcup_{\PartyIndex=1}^{\PartyIndex=\NumPartys}\Dataset_{\PartyIndex}$ is the union of the datasets held by the clients, $\Clusters = \{\ClusterI{1}, \ClusterI{2}, ..., \ClusterI{\NumClusters}\}$ are the clusters, and $\Centers = \{\CenterI{1}, \CenterI{2}, ..., \CenterI{\NumClusters}\}$ are the respective centroids of the clusters.

We aim to optimize this objective while formally proving the privacy of the output and intermediate computations. Specifically, the protocol should be secure in the computational model of differential privacy introduced by Mirnov et al.~\cite{mironov_CDP_09} and extended to the multi-party setting by Humphries et al.~\cite{humphries2022}. We further extend this model to the approximate DP setting by incorporating the failure probability $\delta$.
\begin{definition}[IND-CDP-MPC~\cite{humphries2022, mironov_CDP_09}]\label{def:sec_model}
A multi-party protocol $\Pi$ for computing function $f$ satisfies $(\epsilon(\lambda),\delta)$- indistinguishable computationally differential privacy (IND-CDP-MPC) if for every probabilistic polynomial time (in $\lambda$) adversary $A$ with input dataset $D_A$, and for neighbouring datasets $D,D'$ belonging to the honest parties (i.e. $D,D' = \cup_{i \setminus A} D_i$),
\begin{align*}
&\Pr[\mathcal{A}(\textsc{view}_A^{\Pi}(D_A, D))=1] \\
\leq &\exp(\epsilon) \cdot \Pr[\mathcal{A}(\textsc{view}_A^{\Pi}(D_A, D'))=1] + negl(\lambda) + \delta.
\end{align*}
where $\mathcal{A}$ is the function representing the adversary's decision on whether the dataset was $D$ or $D'$ based on their $\textsc{view}_A^\Pi$, which is the transcript of all messages observed by adversary $A$ during execution of protocol $\Pi$.
$negl(\lambda)$ is a negligible function decreasing faster than any inverse polynomial in $\lambda$. Likewise, the definition holds for every other party's view of neighbours $(D, D')$.
\end{definition}

Intuitively, a protocol that satisfies IND-CDP-MPC \emph{securely} simulates a TTP executing a central DP mechanism.
Specifically, even after observing the output and intermediate computations, any (computationally bounded by a polynomial in $\lambda$) party should not learn more about any other party's local dataset than what could be learned from the differentially private leakage of the central DP mechanism itself.

We make the following assumptions when designing our protocol:
\begin{itemize}
    \item The existence of an honest-but-curious service provider \Server~to assist with multiparty computations. This party should be oblivious to all inputs and results, i.e., the service provider learns nothing about the client's input or output during the execution of the protocol (a weaker assumption than the local and shuffle models). We also discuss alternatives to this assumption in Appendix~\ref{sec:extensions}.
    \item Clients share a common secret used as a Pseudorandom Number Generator (PRNG) seed and do not collude with the server. This secret can be established through standard public-key schemes (e.g., Bonawitz et al.~\cite{Bonawitz2017} uses Diffie-Hellman key exchange~\cite{diffie1976new}) or TLS-secured channels. 
\end{itemize}

\subsection{Comparison to other models}\label{sec:model_comparison}
In the \textbf{local DP model}, each client independently perturbs their data before sharing it. 
Since privacy is guaranteed at the source, there is no need to trust an aggregator or rely on any cryptographic protocols.
However, it significantly reduces accuracy because noise is added individually to each data point rather than to aggregated statistics.

The \textbf{central DP model}, in contrast, involves a trusted third-party (TTP) aggregator that collects raw data from clients and applies noise only to aggregated results. 
This significantly improves accuracy over the local model since the sensitivity of aggregated statistics is lower than individual points. 
However, this model requires strong trust assumptions, as the aggregator has full access to the clients' raw data.

The \textbf{shuffle DP model} offers a compromise between local and central DP by introducing a non-colluding semi-trusted shuffler that permutes messages before aggregation.
This shuffling enhances privacy, allowing clients to inject less noise compared to local DP while still not requiring trust in the aggregator.
Although accuracy improves compared to local DP, the shuffle model remains inherently less accurate than the central DP model.
Additionally, it necessitates either trusting the shuffler to remain non-colluding or employing an oblivious shuffle protocol.

On the other hand, the \textbf{IND-CDP-MPC model} ``replaces'' the trusted aggregator from central DP with a secure multi-party computation (MPC) protocol that simulates the central mechanism.
Therefore, it retains the accuracy of the central DP model while operating under a local model of trust (only revealing noised data).
This allows for a significantly better privacy vs. utility trade-off than the local and shuffle DP models, and is therefore the focus of this work.

The \textbf{exact MPC model} allows multiple parties to jointly compute a function over their inputs while keeping those inputs private.
The model assumes that no information beyond the final  (non-private) output is revealed.
However, without protecting the output with DP, the result is vulnerable to reconstruction attacks that can leak the entire dataset in the worst case.
Thus, we focus on the IND-CDP-MPC model.

\section{Related Work}\label{sec:rel_work}
\subsection{Secure Exact \km}

We first discuss secure computation based approaches to \km that do not preserve the privacy of the output.
This literature shows a trade-off between privacy and efficiency.
Approaches that prioritize efficiency often compromise privacy by allowing the leakage of intermediate computations such as sums, counts, or centroids, which can be exploited to infer sensitive data~\cite{Gheid2016, Jiang2020}. 
For instance, Gheid et al.~\cite{Gheid2016} rely on computing secure sums for aggregating sums and counts across clients.
While this is very efficient in practice, it reveals aggregate intermediate values to all clients, which is not considered secure~\cite{Hegde2021}.
As we discuss in Appendix~\ref{sec:extensions}, naively fixing this leakage would require a significant degradation in utility or performance.
Another approach by Jiang et al.~\cite{Jiang2020} employs homomorphic encryption and garbled circuits to conceal sums and counts.
However, it still reveals the intermediate centroids to all clients in the clear.
Many other works in this category are explained in the SoK of Hedge et al.~\cite{Hegde2021}.

Conversely, secure approaches that do not leak intermediate computations employ heavy cryptographic techniques, like fully homomorphic encryption, leading to significant performance degradation~\cite{Bunn2007, Rao2015, Jschke2018, Kim2018}.
Bunn and Ostrovsky~\cite{Bunn2007} utilize arithmetic secret sharing alongside homomorphic encryption, with high-performance overheads. 
Homomorphic encryption has also been used by Rao et al.~\cite{Rao2015}, J{\"a}schke et al.~\cite{Jschke2018}, and Kim et al.~\cite{Kim2018}, and consistently results in prohibitively high computation time due to the computation demands of the encryption. 
The work of Mohassel et al.~\cite{Mohassel2020} is a notable exception, offering a scheme that is significantly faster than previous state-of-the-art methods.
By using optimized $2$-party computation primitives,  batched oblivious transfers, and garbled circuits, they achieve a speedup of five orders of magnitude over J{\"a}schke et al.~\cite{Jschke2018}. 
However, the runtime is still on the order of minutes; and the communication size is on the order of gigabytes.
In contrast, our approach is as efficient as the protocols that leak intermediate computations, while protecting the privacy of the output and intermediate computations with DP.

\subsection{Differentially Private \km}

DP \km algorithms aim to address output privacy by introducing noise during clustering, hiding individual data contributions, and ensuring a formal privacy guarantee.
Research in DP \km is divided into central and local models.

\subsubsection{Central DP} The central model operates with a trusted curator who collects data for analysis, hiding individuals' information in the output, but not from the curator itself. 
Among the central approaches, DP-Lloyd introduces noise during each centroid update, by adding Laplace noise to both the sum (numerator) and the counts (denominator) when computing the arithmetic mean of points within a cluster~\cite{Blum2005}.
DP-Lloyd was later implemented in the PinQ framework~\cite{Mcsherry2009} with a fixed number of iterations.
Dwork~\cite{Dwork2011} extended this framework to allow for an arbitrary number of iterations, allocating an exponentially decreasing privacy budget to each iteration, although it was noted that utility degrades beyond a certain number of iterations due to increasing noise levels.
The state-of-the-art in DP-Lloyd was achieved by Su et al.~\cite{Su2016}, who included an error analysis to determine the optimal number of iterations, optimized the splitting of privacy budget between sum and count queries, and introduced a ``sphere-packing" centroid initialization method that significantly improves clustering quality.

The Sample and Aggregate Framework (\textit{SaF}) partitions the dataset into multiple subsets, upon which the non-private Lloyd's algorithm is executed independently~\cite{Nissim2007}. 
The resulting centroids from each subset are then aggregated using a standard DP mechanism. 
SaF was later implemented within the GUPT system~\cite{Mohan2012}. 
However, experiments conducted by Su et al.~\cite{Su2016} suggest that DP-Lloyd consistently outperforms SaF across various synthetic and real datasets.
Synopsis-based methods take a different approach by first generating a synopsis of the dataset using a differentially private algorithm, then applying the \km clustering algorithm on this synopsis.
Qardaji et al.~\cite{Qardaji2012} proposed a synopsis technique for 2D datasets, which was extended and optimized by Su et al.~\cite{Su2016} to higher dimensions. 
However, this approach only outperforms DP-Lloyd in datasets with less than three dimensions~\cite{Su2016}.
Another, more theoretical, line of work focuses on minimizing approximation error bounds, without implementation or empirical evaluation~\cite{Feldman2009,Feldman2017,Balcan2017,Nissim2018,Stemmer2018,Ghazi2020}. 
The work of Ghazi et al.~\cite{Ghazi2020} concludes this long line of work by achieving the same approximation factor as the best non-private algorithms.
However, these methods suffer from superlinear running times, making them impractical for large datasets, where privacy is most important.
\subsubsection{Local DP} In contrast to central DP, the local model operates under the premise that no trusted curator is available.
This necessitates that individuals randomize their data locally before it is aggregated, which reduces utility by introducing a significantly higher level of noise (proportional to the number of data points). 
Several works have been proposed in this area~\cite{Nissim2018, Stemmer2021, Chang2021}, most of which focus on the theoretical error bounds rather than practical applications.
The most efficient and accurate local DP protocol by Chang et al.~\cite{Chang2021} could be extended to our setting, by instantiating the aggregation oracles with our MSA protocol.
This would significantly improve the accuracy of their protocol.
However, this would introduce high communication complexity ($O(\NumPartys \sum_{\PartyIndex=1}^{\NumPartys}{\DatasetSize_{\PartyIndex}})$), since the complete net-tree would need to be aggregated across all parties.
In contrast, our protocol transmits only centroids and counts per iteration, leading to a complexity of $O(\NumPartys \NumClusters \NumIter)$.
Another line of work studies private aggregation (a building block for federated DP \km) in local and shuffle DP models~\cite{Ghazi2020b,Balle2020}. 
However, these works cannot overcome the utility gap stemming from the noise being added to each data point.
Finally, Li et al.~\cite{Li2023} attempt to bridge the gap between these central and local models in the vertically federated setting.
Their approach relies on an untrusted curator aggregating noisy local centers and membership encodings from the clients to generate a private synopsis.
However, this approach is specifically tailored for the vertical setting and cannot be applied to our configuration.
Our work yields utility that is even better than state-of-the-art work in the central model~\cite{Su2016} while obtaining the trust assumptions of the local model and operating in the federated setting.

\section{Distributed DP \km}\label{sec:main_algo}
First, we overview our modified DP algorithm using the maximum radius constraint and relative updates.
Then, we describe \ourProtocol, including the initialization, assignment, and update steps.
Finally, we provide our security theorem, and discuss how we choose the various parameters.

\subsection{Radius Constrained DP \km}\label{sec:constrained_k_means}
In this section, we describe how we improve the clustering utility over the work of Su et al.~\cite{Su2016}, a contribution that is independently applicable to the central model of DP.
We observe that despite being common in the non-private literature~\cite{bennett2000constrained,Levy-Kramer_k-means-constrained_2018}, neither federated exact \km nor DP \km algorithms in the literature apply any constraints to the \km~objective.
In this work, we focus on a type of constraint not studied in the non-private literature: bounding the radius of each cluster.
This constraint is particularly useful in the context of DP \km, as it allows a tighter sensitivity bound when computing cluster updates relative to the previous centroid.

In the DP \km literature~\cite{Blum2005,Mcsherry2009, Dwork2011, Su2016}, a typical private clustering algorithm is to run \lloydsProtocol's algorithm, with DP noise added to the update step.
The DP mean is computed by two DP queries: one for the sum and one for the count.
The sum is then divided by the count as a post-processing step.
The challenge in these protocols is that the sensitivity of the sum is bounded by the domain size.
If the datapoints $\datapointi{\DatasetIndex}$ are contained in $[-\bound, \bound]^\Ddim$, then the $l_2$ sensitivity of the sum is $\bound \sqrt{\Ddim}$.
This is because, with no other constraints, there exists a worst-case data point that could add $\bound$ to the sum in all dimensions, regardless of the cluster it is assigned to.

In this work, we modify the \km~objective such that in addition to minimizing the WCSS (Eqn \ref{eqn:wcss_obj}), a constraint on the maximum radius of each cluster must also be satisfied. Specifically, no data point in a cluster can be more than \MXD~away from the cluster's centroid. Thus, the objective becomes to minimize (Eqn \ref{eqn:wcss_obj}) s.t.:

\begin{equation}\label{eq:maxdist_constraints}
    ||\datapointi{\DatasetIndex} - \CenterI{\ClusterID}||_2 \leq \MXD \quad \forall \datapointi{\DatasetIndex} \in \ClusterI{\ClusterID}, \forall \ClusterID
\end{equation}

To enforce this constraint, we modify the assignment step to not assign a data point to a cluster if it is more than \MXD~away from its nearest centroid.
Any unassigned data points are discarded and, therefore, do not contribute to any sum or count query.
Intuitively, bounding the maximum radius of each cluster limits how much a cluster can move by adding or removing a data point since a worst-case data point must be close to a cluster's centroid to be factored into the mean.
However, simply bounding the radius of the cluster does not tighten the sensitivity of the sum query.
This is because the sensitivity analysis must consider the worst-case cluster.
If a cluster is within \MXD~of the domain boundary, then the $\ell_2$ sensitivity of the sum query is still $\bound\sqrt{\Ddim}$.

To realize the reduction in sensitivity from bounding the radius, we must modify the update step as well as the assignment step.
Instead of perturbing the sum of the data points themselves, we perturb the difference between the data points and their assigned cluster centroid.
Specifically:
\begin{equation}
    \Diff{\ClusterID}^{(\IT)} = \sum\limits_{\datapointi{\DatasetIndex} \in \ClusterI{\ClusterID}^{(\IT)}} \datapointi{\DatasetIndex} - \NoiseCenterI{\ClusterID}^{(\IT-1)}
\end{equation}
Essentially, we are computing the ``updates'' to the centroids, rather than the centroids themselves.
This simple change has a significant impact on the sensitivity of the sum query.
The sensitivity of the relative sum is now bounded by the maximum radius parameter $\MXD$:
\begin{restatable}{theorem}{RadiusSensitivity}\label{thm:maxdist_sens}
    If the constraint \eqref{eq:maxdist_constraints} is satisfied, then:
    \begin{equation}
        \Delta^{\Diff{}} = \max\limits_{\Dataset, \Dataset' \in \mathcal{D}} ||\Diff{\ClusterID}^{(\IT)}(D) - \Diff{\ClusterID}^{(\IT)}(D')||_2 \leq \MXD
    \end{equation}
    for all clusters $\ClusterID \in [\NumClusters]$, and iterations $\IT \in [\NumIter]$.
\end{restatable}
We prove this in Appendix~\ref{app:proof_maxdist_sens}. The intuition is that a worst-case data point only contributes its distance to the centroid (which is bounded by $\MXD$), rather than its distance to the origin.

Our maximum distance constraint naturally bounds the $l_2$ sensitivity as it bounds the Euclidean distance, and thus, Gaussian noise is a natural choice that also scales more efficiently to higher dimensions.
It also allows for a simple and tight privacy analysis using Gaussian-DP~\cite{dong22_gdp}.
We use the analysis of Balle and Wang~\cite[Algorithm 1]{balle18_gaussian} to obtain the noise multiplier $\sigma$.
We divide the privacy budget between the relative sum and the count following our analysis in Section~\ref{sec:dp_params} as:
\begin{equation}\label{eq:dp_params}
    \sigma^{\Diff{}} = \frac{\sigma\sqrt{1 +\sqrt{4\Ddim}}}{\sqrt[4]{4\Ddim}} \quad \sigma^{\Cnt{}} = \sigma\sqrt{1 +\sqrt{4\Ddim}}
\end{equation}
We then add noise as follows:
\begin{equation}\label{eq:sum_noise}
    \NoiseDiff{\ClusterID}^{(\IT)} = \Diff{\ClusterID}^t + \DPNoise^{\Diff{}} \quad \text{where} \quad \DPNoise^{\Diff{}} \sim \mathcal{N}(0, (\sigma^\Diff{})^2 \MXD^2 \NumIter)
\end{equation} 
\begin{equation}\label{eq:count_noise}
    \NoiseCnt{\ClusterID}^{(\IT)} = \Cnt{\ClusterID}^t + \DPNoise^{\Cnt{}} \quad \text{where} \quad \DPNoise^{\Cnt{}} \sim \mathcal{N}(0, (\sigma^\Cnt{})^2 \NumIter).
\end{equation}
We prove in Section~\ref{sec:privacy_proof} that adding noise in this way satisfies DP.
After adding noise, we compute the new centroid as:
\begin{equation}\label{eq:unshifted_centroid}
    \NoiseCenterI{\ClusterID}^{(\IT)} = \frac{\tilde{\Diff{\ClusterID}}^{(\IT)}}{\tilde{\Cnt{\ClusterID}}^{(\IT)}} + \NoiseCenterI{\ClusterID}^{(\IT-1)}.
\end{equation}

Without DP noise, this yields an equivalent cluster update. 
\begin{eqnarray}
    \CenterI{\ClusterID}^{(\IT)} 
    &=& \frac{\Sm{\ClusterID}^{(\IT)} - \Cnt{\ClusterID}^{(\IT)}\CenterI{\ClusterID}^{(\IT-1)}}{\Cnt{\ClusterID}^{(\IT)}} + \CenterI{\ClusterID}^{(\IT-1)} = \frac{\Sm{\ClusterID}^{(\IT)}}{\Cnt{\ClusterID}^{(\IT)}}
\end{eqnarray}
where $\Sm{\ClusterID}^{(\IT)} = \sum_{\datapointi{\DatasetIndex} \in \ClusterI{\ClusterID}^{(\IT)}} \datapointi{\DatasetIndex}$ is the sum of the data points. With the DP noise, we get an additional noise (error) term in the sum compared to noising the sum directly:
\begin{eqnarray}
    \NoiseSm{\ClusterID}^{(\IT)} &=& \NoiseDiff{\ClusterID}^{(\IT)} + \NoiseCnt{\ClusterID}^{(\IT)}\NoiseCenterI{\ClusterID}^{(\IT-1)} \\
    &=& \Sm{\ClusterID}^{(\IT)} - \Cnt{\ClusterID}^{(\IT)}\NoiseCenterI{\ClusterID}^{(\IT-1)} + \DPNoise^{\Diff{}} + (\Cnt{\ClusterID}^{(\IT)}+\DPNoise^{\Cnt{}})\NoiseCenterI{\ClusterID}^{(\IT-1)}\\
    &=& \Sm{\ClusterID}^{(\IT)} + \DPNoise^{\Diff{}} + \DPNoise^{\Cnt{}}\NoiseCenterI{\ClusterID}^{(\IT-1)}.
\end{eqnarray}
However, as we show in Section~\ref{sec.utility_eval}, this additional error is compensated for by the increase in utility from reducing the sensitivity from $\bound \sqrt{\Ddim}$ to $\MXD$. The maximum radius further decreases error by reducing the effect of outliers, as data points far away from any centroid will not (and, in some cases, should not) be assigned to any cluster. We also use the maximum radius constraint as an additional post-processing constraint, which we call \textit{Radius Clipping}. Namely, if a noisy centroid is computed to be more than $\MXD$ away from the previous centroid, we truncate it to be $\MXD$ away.

\subsection{Overview of \ourProtocol}
In this section, we describe our protocol for federated DP \km, \ourProtocol.
The focus of \ourProtocol~is to create the most efficient and accurate protocol possible under the threat model defined in Section~\ref{sec:prob_state}.
Specifically, we refrain from adding additional noise or computations that would be needed to handle client failures or collusion between clients and the server.
We leave it to future work to adapt our approach to use more resilient aggregation protocols~\cite{Bonawitz2017, Kadhe2020fastsecagg, So2021turbo, Fereidooni2021safelearn}.

\subsubsection{Protocol Intuition}

A naive IND-CDP-MPC implementation of Lloyd's algorithm would use an end-to-end secure protocol (e.g.,~\cite{Mohassel2020}) to compute the exact centroids in every iteration, then add DP noise using a secure computation.
This entails gigabytes of communication and tens of minutes of runtime as we show in Section~\ref{sec:eval}.

Instead, our key insight is that the tightest DP analysis for Lloyd's algorithm is an (adaptively)-compositional proof~\cite{Su2016,Blum2005} that assumes that the (perturbed) intermediate computations are published in each iteration.
Therefore, revealing these intermediate DP updates does not violate securely simulating the TTP, as required by the IND-CDP-MPC model \footnote{This insight would not apply to any future DP analyses that achieve privacy amplification by specifically hiding intermediate computations. Our approach would still be applicable, but at would not benefit from this amplification.}.
This allows us to compute divisions and assignments locally, rather than in a secure computation protocol.

We note that leaking the intermediate computations is not possible in the exact MPC model.
Thus our approach is not applicable to the exact MPC model (where the output is not private).
In this paper, we demonstrate that, by working within the computational differential privacy framework, the state-of-the-art DP version of Lloyd's algorithm can be computed significantly more efficiently (five orders of magnitude faster) in MPC than its non-private counterpart.
\subsection{Algorithm Description}
Algorithm~\ref{alg:Local_update} overviews the protocol from the perspective of a single client.
Following the outline of Lloyd's algorithm, the following sections describe how we design each of its main steps: initialization, assignment, and update.

\subsubsection{Initialization: Sphere Packing}
We modify the initialization so that it can be carried out in a federated manner.
We employ the sphere packing initialization approach of Su et al.~\cite{Su2016} as it was shown to outperform random initialization.
The sphere packing approach is data-independent and thus does not use any privacy budget.
The process can be outlined as follows:
\begin{enumerate}
    \item Initialize a radius parameter $a$.
    \item For $\ClusterID \in  \{1, 2, ..., \NumClusters\}$, generate a point $\CenterI{\ClusterID}$ such that it is at least of distance $a$ away from the domain boundaries and at least of distance $2a$ away from any previously chosen centroid. If a randomly generated point does not meet this condition, generate another one.
    \item If, after $100$ repeated attempts, it is not possible to find such a point, decrease the radius $a$ and repeat the process.
  \end{enumerate}
The radius $a$ is determined via a binary search
to find the maximum $a$ that allows for the generation of $\NumClusters$ centroids.
In Algorithm~\ref{alg:Local_update}, each client independently calls the \texttt{Initialization} function (Line~\ref{line:init}) with the same random seed, which results in each client starting with the same centroids.

\subsubsection{Assignment: Radius Constrained \km}

In Line~\ref{line:min_cluster}, each client locally computes the closest cluster to each of their data points.
If the data point is within $\MXD$ of the cluster's centroid, it is assigned to that cluster (Line~\ref{line:assignment}).
Any points further than $\MXD$ from the cluster's centroid are not assigned to any cluster. 
We discuss how we set $\MXD$ in Section~\ref{sec:radius}.

\begin{algorithm}[t]
    \caption{\ourProtocol~from $\Party{\PartyIndex}$'s perspective}\label{alg:Local_update}
    \begin{algorithmic}[1]
    \Statex \textbf{Inputs: }Local Dataset $\Dataset_\PartyIndex$.
    \Statex \textbf{Output:} Cluster Centres $\NoiseCenters^{(\NumIter)}$
    \State $\NoiseCenters^{(0)} =$ \Call{Initialization}{$seed$}\label{line:init} 
    \For{$t \in \{1:\NumIter\}$}
    \State // \textbf{Assignment Step}
    \For{$\datapointi{\DatasetIndex} \in \Dataset_\PartyIndex$}
    \State $j' = \argmin\limits_{j'\in[\NumClusters]} ||\datapointi{\DatasetIndex} - \NoiseCenterI{j'}^{(\IT-1)}||_2$\label{line:min_cluster}
    \If{$||\datapointi{\DatasetIndex} - \NoiseCenterI{j'}^{(\IT-1)}||_2 < \MXD$}
    \State $\ClusterI{\PartyIndex j'}^{(\IT)} \gets \datapointi{\DatasetIndex}$\label{line:assignment}
    \EndIf
    \EndFor
    \State // \textbf{Local Update}
    \For{$\ClusterID \in \{1, \dots, \NumClusters\}$} 
    \State Compute $\bar{\Diff{}}_{\PartyIndex\ClusterID}^{(\IT)} = \sum\limits_{\datapointi{\DatasetIndex} \in \ClusterI{\PartyIndex\ClusterID}^{(\IT)}} \datapointi{\DatasetIndex} - \NoiseCenterI{\ClusterID}^{(\IT-1)}$\label{line:local_diffs}
    \State Compute $\bar{\Cnt{}}_{\PartyIndex\ClusterID}^{(\IT)} = |\ClusterI{\PartyIndex\ClusterID}^{(\IT)}|$\label{line:local_counts}
    \EndFor
    \State // \textbf{Global Update}
    \State $\widehat{\boldsymbol{U}} = $ \Call{GlobalMSA}{$\bar{\Diff{}}_{\PartyIndex\ClusterID}^{(\IT)}$, $\sigma^\Diff{}$, seed, $\IT$}\label{line:global_diffs}
    \State $\widehat{\boldsymbol{C}} = $ \Call{GlobalMSA}{$\bar{\Cnt{}}_{\PartyIndex\ClusterID}^{(\IT)}$, $\sigma^\Cnt{}$, seed, $\IT$}\label{line:global_counts}
    \State // \textbf{Post Process Result}
    \For{$\ClusterID \in \{1, \dots, \NumClusters\}$}
    \State $\widehat{\mu}_{\ClusterID} = \frac{\widehat{\boldsymbol{U}}}{\widehat{\boldsymbol{C}}} + \NoiseCenterI{\ClusterID}^{(\IT-1)}$\label{line:compute_centroid}
    \EndFor
    \State $\NoiseCenters^{(\IT)} =$  \Call{Fold}{$\widehat{\boldsymbol{\mu}}$}\label{line:fold}
    \EndFor
    \end{algorithmic}
\end{algorithm}

\subsubsection{Local Update}

The client has already locally assigned each of their data points to a cluster, following the constraints in the previous step (Line~\ref{line:assignment}).
Next, they compute the relative sum and the count for each cluster using their local dataset.
We call this the local update step.
The output of the relative sum (Line~\ref{line:local_diffs}) and the count (Line~\ref{line:local_counts}) are two matrices of dimensions $(\NumClusters\times\Ddim)$ and $(\NumClusters\times 1)$ respectively.

\subsubsection{Global Update: Masked Secure Aggregation}
We modify the update step to release a perturbed version of relative sums and the counts to each client under the CDP security model.
By publishing the previous noisy centroids, the assignment and local update computation can be performed locally by each client instead of using a secure computation protocol.
In the global update step (Lines~\ref{line:global_diffs} and~\ref{line:global_counts}), we privately aggregate (sum) the local updates of all clients, perturb the result, and reveal the noisy global relative sums and counts for the next round.

We describe our aggregation protocol in Figure~\ref{fig:MSA_protocol}.
We call this protocol Masked Secure Aggregation (MSA).
In MSA, similar to regular secure aggregation protocols~\cite{Bonawitz2017}, a service provider \Server~securely aggregates values from \NumPartys~clients, all while being oblivious to every client's contributions. 
However, in MSA, the server is also oblivious to the result of the aggregation operation; the server only acts as an aggregator (who also adds DP noise) and is oblivious to both the input and output.
Each client $\Party{\PartyIndex}$ possesses a private value (a matrix $\VV{\PartyIndex} \in \mathbf{R}^{\NumClusters \times \Ddim}$, where $\mathbf{R}$ could be $\mathbb{Z}_{2^{64}}$).
These clients aim to collectively compute the element-wise sum of their private matrices: $\boldsymbol{\VV{}}=\sum_{\PartyIndex=1}^{\NumPartys}{\frac{\VV{\PartyIndex}}{\NumPartys}}$.
In addition to the client's values, we assume the noise scale $\sigma$ and a random $seed$ (that all clients know) are also provided.

\emph{Clients Send Data.} The first step in the protocol is for each of the $\NumPartys$ clients to generate a random mask set.
The random mask set is composed of \NumPartys~random matrices of the same dimensions as the client items: $\{\MM{1}, \MM{2}, \dots \MM{\NumPartys}\}$ where each value is sampled uniformly from $\mathbf{R}$.
Because all the clients have access to a shared seed and $PRNG$ they can each compute the entire set locally.
In Line~\ref{line:get_mask}, each client samples the entire set and sums it (in Line~\ref{line:agg_mask}) to get the global mask matrix $\boldsymbol{\MS}$ (which will be used to decrypt later).

In Line~\ref{line:fixed_point}, each private value is converted to a fixed-point format, \fixedVV{\PartyIndex}, by scaling it up with a power-of-2 scale factor, \SF. Formally, $\fixedVV{\PartyIndex} = \round{\VV{\PartyIndex} \times \SF}$. 
All fractional values are rounded to the nearest representable value in this fixed-point representation. The choice of the scale factor, \SF, determines the precision of the representation\footnote{A larger \SF~allows for greater precision but also reduces the number of bits available for the integer part of the number, which might cause overflows. A workaround is to increase the bit-width of the operations, which increases the computational load and the communication cost.}. In practice, we empirically choose $\Prec=16$.

Then, in Line~\ref{line:client_enc}, each client encrypts (masks) their input $\fixedVV{\PartyIndex}$ by adding their share of the mask $\MM{\PartyIndex}$.
Finally, each client sends their masked values to the server.
This step masks the actual value \VV{\PartyIndex} from the server because \MM{\PartyIndex} acts as a one-time pad. 
Note that this masked value is never sent to other clients, as they would be able to unmask it easily.

\begin{figure}[t]
\noindent\maxsizebox{\columnwidth}{!}{
  \begin{varwidth}{1.4\columnwidth}
\game[linenumbering, skipfirstln, mode=text]{\textbf{GlobalMSA}($\VV{\PartyIndex},\sigma, seed, \IT$)}{%
\textbf{Clients} $\PartyIndex \in \{1,\ldots,\NumPartys\}$ \<\< \textbf{Server $\Server$} \pcskipln\\
// Each client $\PartyIndex$ masks \\
\label{line:get_mask}$\{\MM{1},...,\MM{\NumPartys}\} = PRNG(seed, \IT)$   \\
\label{line:agg_mask}$\MS \equiv \sum_{\PartyIndex=1}^{\NumPartys} \MM{\PartyIndex} $\<\< \\
\label{line:fixed_point}$\fixedVV{\PartyIndex} \coloneqq \round{\VV{\PartyIndex} \times \SF}$ \<\< \\
 \label{line:client_enc}$Enc_{\PartyIndex}(\VV{\PartyIndex}) \coloneqq \fixedVV{\PartyIndex} + \MM{\PartyIndex}$ \<\< \\
\< \sendmessageright*[1.25cm]{Enc_{\PartyIndex}(\VV{\PartyIndex})} \< \pcskipln \\
\<\< $\text{// Compute sum}$\\
\label{line:server_sums}\<\<   $\fixedVV{}+ \MS{} \equiv \sum_{\PartyIndex=1}^{\NumPartys{}}{Enc_{\PartyIndex}(\VV{\PartyIndex})} $ \pcskipln \\
\<\< $\text{// Add noise value}$\\
\<\< $\DPNoise \sim \mathcal{N}(0,\sigma^{2})$\\
\<\< $\tilde{\DPNoise} = \round{\DPNoise \times \SF}$\\ 
\<\< $Enc(\VV{}+\DPNoise) \equiv \fixedVV{}+\MS{}+\tilde{\DPNoise} $\\
\< \sendmessageleft*[1.25cm]{Enc(\VV{}+\DPNoise)} \< \\
$\VV{}+\DPNoise \approx (Enc(\VV{}+\DPNoise) - \MS)/\SF$\<\<}
\end{varwidth}
}
\caption{Global Masked Secure Aggregation Protocol with $\NumPartys$ Clients.}
\label{fig:MSA_protocol}
\end{figure}

\emph{Server Aggregates Data.}
Upon receiving the masked matrices, the server first sums over each client's contribution in Line~\ref{line:server_sums}.
This yields the masked global sum $\boldsymbol{\VV{}}+ \boldsymbol{\MS{}}$.
The server then samples from a zero mean Gaussian distribution with standard deviation equal to the supplied $\sigma$ for each entry in the result matrix.
This noise must also be converted to fixed-point by computing $\tilde{\DPNoise} = \round{\DPNoise \times \SF}$. We show why this preserves differential privacy in Appendix~\ref{app:proof_quant}.
Finally, the server adds the noise $\tilde{\DPNoise}$ to the sum and broadcasts it to the clients.

\emph{Client Unmasks Data.}
Upon receiving the result from the server, each client must unmask the result.
They do this by simply subtracting the mask $\MS{}$ that they computed in Line~\ref{line:agg_mask}.
After unmasking, each client must scale down the result to retrieve the correct answer. This is done by dividing the unmasked sum by the scale factor, \SF, to reverse the initial scaling operation. 

\subsubsection{Post-processing}
In Algorithm~\ref{alg:Local_update}, each client locally post-processes the results to obtain the centroids for the next iteration.

First, each party divides the relative sum by the count and shifts the result by the previous centroid to get the new centroid (Line~\ref{line:compute_centroid}) following Eqn~\ref{eq:unshifted_centroid}.
Then, if the new centroid is more than $\MXD$ away from the previous centroid, we truncate it to be $\MXD$ away.

Finally, we apply a post-processing step to the centroids to ensure they remain within the domain.
A naive post-processing strategy is to simply truncate out-of-bounds centroids to the boundary.
However, in practice, we find (in Appendix~\ref{app:ablation}) that \emph{folding}~\cite{diffprivlib} the value (reflecting it over the boundary) gives better utility.
More formally, the operation folds a value $x$ into the range $[-\bound, \bound]$ through modular arithmetic. We first compute $(x + \bound) \bmod(2\bound)$, then if this result exceeds $\bound$, we reflect it about $\bound$ by subtracting it from $2\bound$. Finally, we return the value to the target range by subtracting $\bound$. This approach creates a periodic folding pattern that naturally reflects values across the boundaries while preserving distances from the nearest boundary point.
We show in Appendix \ref{app:ablation} that folding improves the algorithm's utility.

\subsection{Privacy Analysis}\label{sec:privacy_proof}
\ourProtocol~allows for a central party to add noise, retaining the utility of the central model of DP. However, since the server only interacts with masked values, and the clients can only observe a differentially private view of the final (noised) centroids, the protocol provides a level of privacy akin to that of the local model of DP. 

We state the end-to-end security Theorem of our algorithm and defer the proof to Appendix~\ref{app:proof_sec_proof}.

\begin{restatable}{theorem}{SecProof}\label{thm:sec_proof}
Algorithm~\ref{alg:Local_update} ensures $(\epsilon(\lambda),\delta)$-IND-CDP-MPC in the presence of a semi-honest, polynomial time adversary who controls at most a single party.
\end{restatable}

\subsection{Error Analysis}\label{sec:dp_params}
To analyze the error of our approach, we follow a similar approximate error analysis as Su et al.~\cite{Su2016}.
The purpose of the analysis is primarily to choose the ratio of the privacy parameters and the number of iterations.
Thus, the analysis makes a series of approximations.
Following Su et al., we consider a single iteration and a single cluster for this analysis. 
To simplify the notation, we omit the cluster index $\ClusterID$ and instead index the variables by the dimension $\DimI$.
We analyze the mean-squared error between the true centroid ($\CenterI{}$) and the differentially-private centroid ($\NoiseCenterI{}$) for one iteration across all dimensions.
\begin{equation}
    MSE(\NoiseCenterI{}^{(\IT)}) = \EX{ \sum\limits_{\DimI=1}^{\Ddim} (\CenterI{\DimI}^{(\IT)} - \NoiseCenterI{\DimI}^{(\IT)})^2 }
\end{equation}
We first expand the following term using the definition of our DP mechanism from Section~\ref{sec:constrained_k_means}:
\begin{eqnarray*}
    \CenterI{\DimI}^{(\IT)} - \NoiseCenterI{\DimI}^{(\IT)} &=& \frac{(\Cnt{}^{(\IT)}+\DPNoise^\Cnt{})\CenterI{\DimI}^{(\IT)}}{\Cnt{}^{(\IT)}+\DPNoise^\Cnt{}} - \frac{\Sm{\DimI}^{(\IT)} + \DPNoise^{\Diff{}}_\DimI + \DPNoise^{\Cnt{}}\NoiseCenterI{\DimI}^{(\IT-1)}}{\Cnt{}^{(\IT)}+\DPNoise^\Cnt{}}\\
    &=& \frac{\DPNoise^{\Cnt{}}(\CenterI{\DimI}^{(\IT)} - \NoiseCenterI{\DimI}^{(\IT-1)}) - \DPNoise^{\Diff{}}_\DimI}{\Cnt{}^{(\IT)}+\DPNoise^\Cnt{}}.
\end{eqnarray*}
Let $\nabla_\DimI$ = $\CenterI{\ClusterID}^{(\IT)} - \NoiseCenterI{\ClusterID}^{(\IT-1)}$. Then, after squaring, we get:
\begin{equation*}
    (\CenterI{\DimI}^{(\IT)} - \NoiseCenterI{\DimI}^{(\IT)})^2 = \frac{(\DPNoise^{\Cnt{}})^2(\nabla_\DimI)^2 - 2\DPNoise^{\Cnt{}}\DPNoise^{\Diff{}}_\DimI\nabla_\DimI +(\DPNoise^{\Diff{}}_\DimI)^2 }{(\Cnt{}^{(\IT)}+\DPNoise^\Cnt{})^2}.
\end{equation*} 
After taking the expectation over all clusters and assuming that $\Cnt{}+\DPNoise^\Cnt{} \approx \DatasetSize/\NumClusters$, following Su et al.~\cite{Su2016}, we get:
\begin{eqnarray*}
    \mse{\NoiseCenterI{\DimI}^{(\IT)}} &\approx& \frac{\NumClusters^3}{\DatasetSize^2}\left(\EX{(\DPNoise^{\Cnt{}})^2}\EX{\nabla_\DimI}^2 + \EX{(\DPNoise^{\Diff{}}_\DimI)^2}\right)\\
    &=& \frac{\NumClusters^3}{\DatasetSize^2}\left(\Var{\DPNoise^{\Cnt{}}}\EX{\nabla_\DimI}^2 + \Var{\DPNoise^{\Diff{}}_\DimI}\right)
\end{eqnarray*}
where, in the first line, the middle term goes is zero as the noise terms are i.i.d. and zero mean.
The second line holds because each $\DPNoise$ is an independent variable with zero mean, and so
\[\E[(\DPNoise^{})^2]=\Var{\DPNoise^{}}-(\E[\DPNoise^{}])^2=\Var{\DPNoise^{}}.\]
Finally, we approximate $\E[\sum_{\DimI=1}^{\Ddim} \nabla_\DimI] \approx \frac{\MXD}{2}$.
We argue that if the data were uniformly distributed in a hypersphere of radius $\MXD$ around the centroid, the expected distance from the centroid would be $\MXD/2$.
This gives the final error term over all dimensions:
\begin{equation}\label{eq:first_error}
    \mse{\NoiseCenterI{}^{(\IT)}} \approx \frac{\NumClusters^3}{4\DatasetSize^2}\left(\Var{\DPNoise^{\Cnt{}}}\MXD^2 + 4\Ddim\Var{\DPNoise^{\Diff{}}}\right)
\end{equation}

\paragraph{Privacy Budget Splitting}
We now use the approximate analysis in Eqn~\ref{eq:first_error} to determine the optimal privacy budget split between the relative sum and the count.
We substitute the variance of the noise terms from Eqn~\ref{eq:dp_params} into Eqn~\ref{eq:first_error}:
\begin{equation}
    \mse{\NoiseCenterI{}^{(\IT)}} \approx \frac{\NumClusters^3\MXD^2\NumIter}{4\DatasetSize^2}\left((\sigma^{\Cnt{}})^2 + 4\Ddim(\sigma^{\Diff{}})^2\right).
\end{equation}
In our privacy analysis, we compute the noise multiplier $\sigma$ such that we achieve $\frac{1}{\sigma}$-GDP. Thus, we need to split the noise multiplier between the relative sum and the count following Theorem~\ref{thm:gdp_comp}:
\begin{equation}
    \sqrt{\left(\frac{1}{\sigma^\Diff{}}\right)^2 + \left(\frac{1}{\sigma^\Cnt{}}\right)^2} = \frac{1}{\sigma}
\end{equation}

We minimize the per iteration error, subject to this constraint, using Lagrange multipliers, which gives:
\begin{equation}
    \sigma^{\Cnt{}} = \sqrt[4]{4\Ddim}\sigma^{\Diff{}}.
\end{equation}
Simply scaling each sigma by this ratio gives Eqn~\ref{eq:dp_params}.
Substituting Eqn~\ref{eq:dp_params} back into the error analysis gives:
\begin{equation}
    \mse{\NoiseCenterI{}^{(\IT)}} = \frac{\NumClusters^3\MXD^2\NumIter\sigma^2 (1+\sqrt{4\Ddim})^2}{4\DatasetSize^2}
\end{equation}

\paragraph{Optimal Number of Iterations}
Using this analysis, we can determine a heuristic for the number of iterations.
Following Su et al.~\cite{Su2016}, we assume that the error in each iteration is less than $0.004\times\bound$.
Re-arranging for $\NumIter$ gives:
\begin{equation}
    \NumIter < \frac{4\DatasetSize^2(0.004)}{\NumClusters^3\MXD^2\sigma^2 (1+\sqrt{4\Ddim})^2}.
\end{equation}
However, as Su et al.~\cite{Su2016} explain, in practice, there is no need to go beyond seven iterations and at least two iterations are needed to gain useful results. Therefore, we truncate this analysis such that the number of iterations is in the range $[2,7]$.

\subsection{Setting the Maximum Radius Parameter}\label{sec:radius}
We derive a dimensionality-aware radius parameter $\MXD$ based on the geometric properties of the feature space.
Consider a $\Ddim{}$-dimensional bounded feature space where each dimension is constrained to the interval $[-\bound{}, \bound{}]$. The domain diagonal $\Ddiag$, representing the maximum possible distance between any two points in this space, is given by:
\begin{equation}
\Ddiag = \sqrt{\Ddim{}} \times (2\bound{})
\end{equation}
When partitioning this space into $\NumClusters{}$ clusters, each cluster occupies a fraction of the total volume. The total volume of the feature space is $(2\bound{})^{\Ddim{}}$. Assuming uniform partitioning, each cluster occupies a volume of $(2\bound{})^{\Ddim{}} / \NumClusters{}$. Consequently, the effective length of each cluster along any dimension can be expressed as:
\begin{equation}
L_c = \frac{2\bound{}}{\NumClusters{}^{1/\Ddim{}}}
\end{equation}
Assuming hypercube-shaped clusters, the maximum distance from a centroid to any point within the cluster (the cluster radius) is half the cluster's diagonal. This can be calculated as:
\begin{equation}
\text{Cluster Radius} = \frac{L_c}{2} \times \sqrt{\Ddim{}} = \frac{\bound{}}{\NumClusters{}^{1/\Ddim{}}} \times \sqrt{\Ddim{}} = \frac{\Ddiag}{2\NumClusters{}^{1/\Ddim{}}}
\end{equation}
Based on this analysis, we propose a heuristic $\eta_{\Ddim{}}$ to constrain the cluster radius in a $\Ddim{}$-dimensional space as:
\begin{equation}
\eta_{\Ddim{}} = \frac{\alpha \Ddiag}{2\NumClusters{}^{1/\Ddim{}}}
\end{equation}
where $\alpha$ is a scaling factor we set to $0.8$ based on experiments using synthetic data in Appendix~\ref{app:ablation}.

In scenarios with low~\Ddim{}~and high~\NumClusters{}, $\MXD$ becomes increasingly restrictive as $\NumClusters{}^{1/\Ddim{}}$ grows larger, providing tighter bounds on cluster radii. However, as dimensions increase, the curse of dimensionality necessitates larger cluster radii to accommodate the exponential growth of volume in the space, $(2\bound{})^{\Ddim{}}$. Our approach takes this into account through $\NumClusters{}^{1/\Ddim{}}$ approaching unity in high dimensions. This dimensional scaling aligns with the intuition that radius constraints are most meaningful in lower-dimensional spaces, where cluster boundaries remain well-defined, whereas, in high-dimensional spaces, the curse of dimensionality renders such constraints increasingly less effective as distance metrics lose their discriminative power.

We consider two distinct approaches to implement the radius constraint. In the ``Constant'' approach, the radius constraint is consistently enforced throughout the clustering process. Alternatively, the ``Step'' approach initializes with a broader constraint of $\Ddiag/2$ at iteration zero and then transitions to the computed radius $\MXD$ for the remaining iterations. This two-phase strategy allows for initial flexibility in centroid placement while gradually imposing stricter constraints for fine-tuning the centroids. We show in Appendix~\ref{app:ablation} that the ``Step'' approach is superior.

\section{Evaluation}\label{sec:eval}
In this section, we provide an extensive evaluation of our work in terms of the following questions:
\begin{enumerate}[label=Q\arabic*]
    \item How does the utility of \ourProtocol~ compare with state-of-the-art in central model~\cite{Su2016} for interactive DP \km? \label{q:util_compare}
    \item How does the utility of \ourProtocol~scale with varying the number of dimensions and number of clusters? \label{q:util_scale}
    \item How does the runtime, communication, and number of rounds for \ourProtocol~compare with the state-of-the-art in federated \km using secure computation~\cite{Mohassel2020}? \label{q:time_compare}
    \item How does the runtime and communication of \ourProtocol~ scale with varying the number of dimensions, clusters, and data points? \label{q:time_scale}
\end{enumerate}
\begin{figure*}[t]
    \centering
    \begin{subfigure}{\textwidth}
        \centering
        \includegraphics[width=0.5\textwidth]{figs/metrics/legend.pdf}
    \end{subfigure}
    \begin{subfigure}{0.19\textwidth}
        \includegraphics[width=\textwidth]{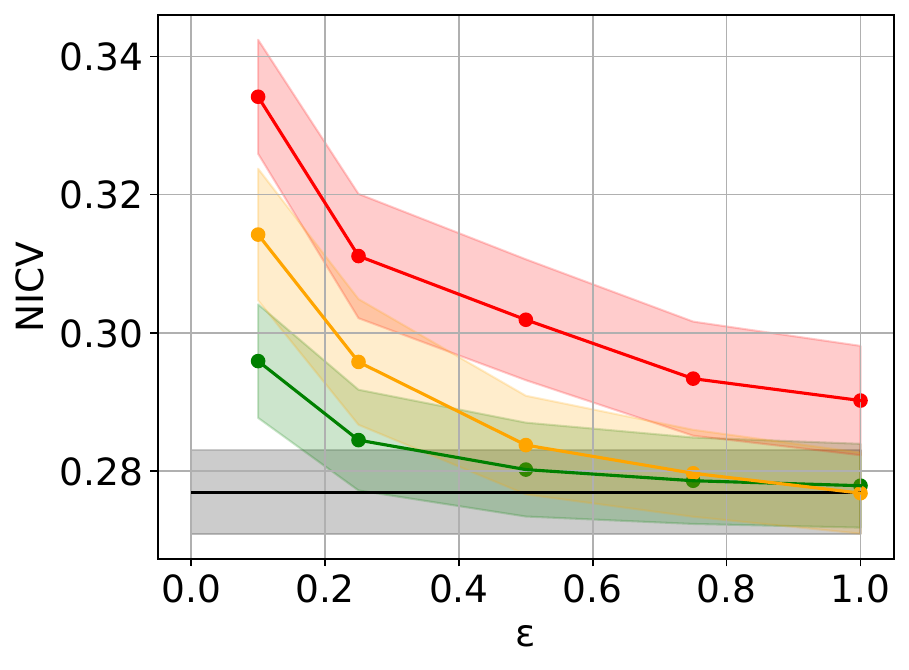}
        \caption{Adult}
        \label{fig:adult}
    \end{subfigure}
    \begin{subfigure}{0.19\textwidth}
        \includegraphics[width=\textwidth]{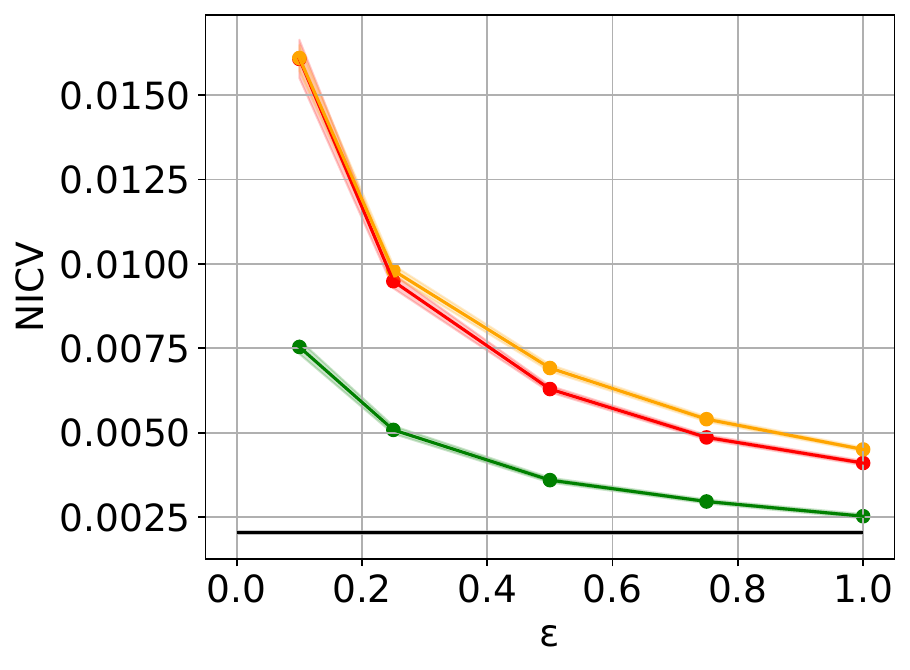}
        \caption{Birch2}
        \label{fig:birch2}
    \end{subfigure}
    \begin{subfigure}{0.19\textwidth}
        \includegraphics[width=\textwidth]{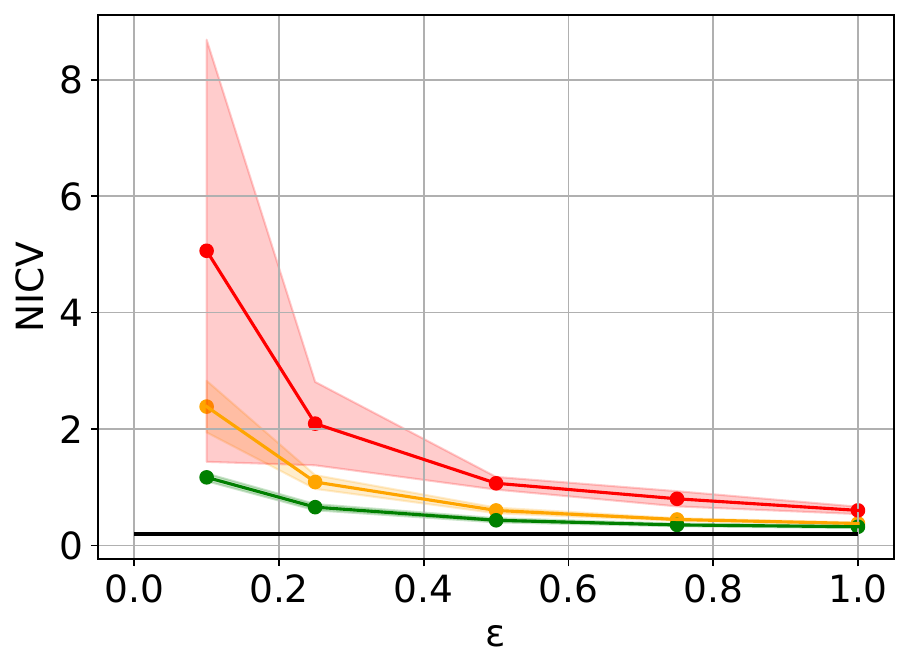}
        \caption{Iris}
        \label{fig:iris}
    \end{subfigure}
    \begin{subfigure}{0.19\textwidth}
        \includegraphics[width=\textwidth]{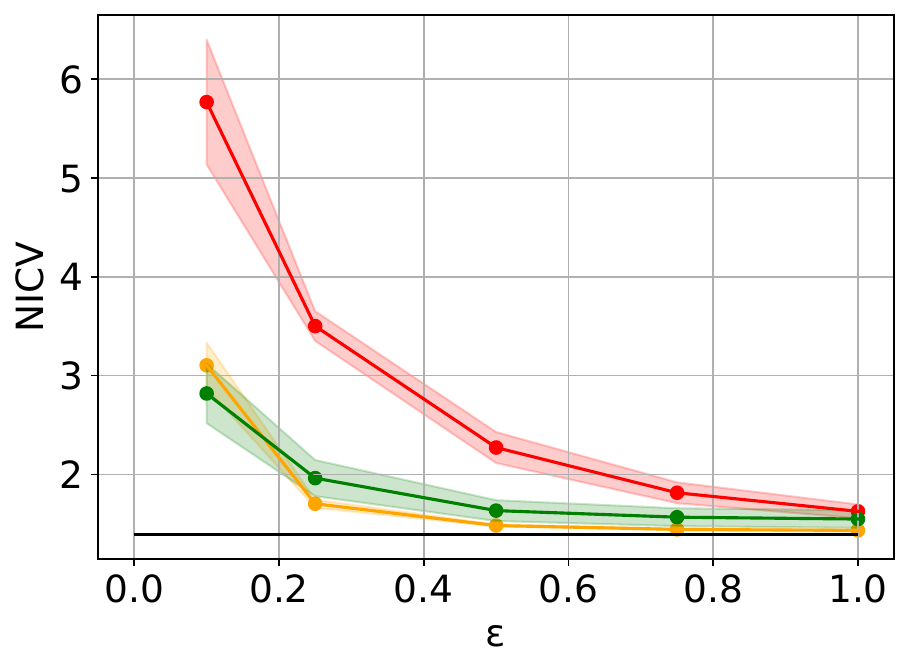}
        \caption{Breast}
        \label{fig:breast}
    \end{subfigure}
    \begin{subfigure}{0.19\textwidth}
        \includegraphics[width=\textwidth]{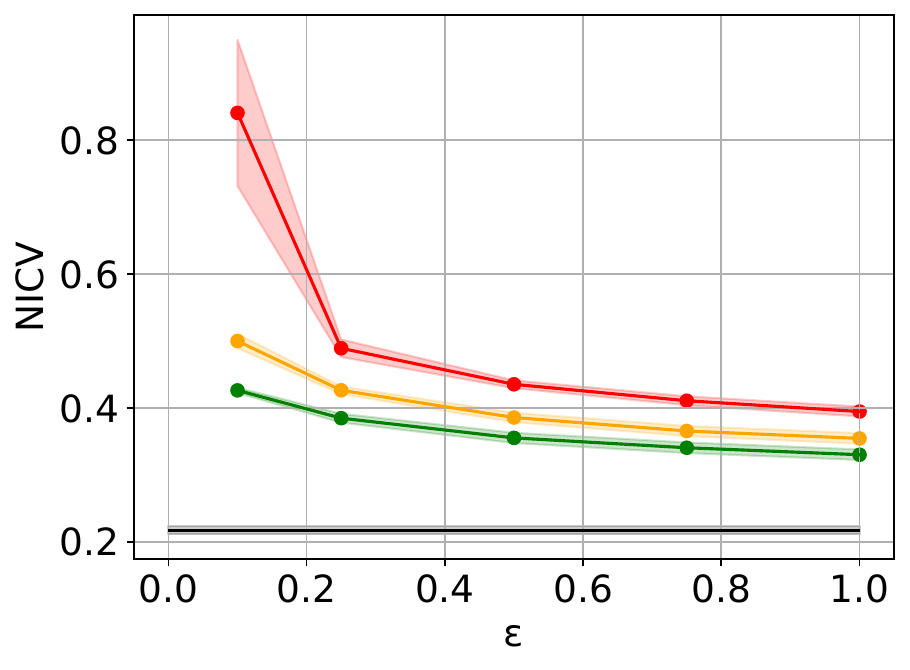}
        \caption{Yeast}
        \label{fig:yeast}
    \end{subfigure}
    \begin{subfigure}{0.19\textwidth}
        \includegraphics[width=\textwidth]{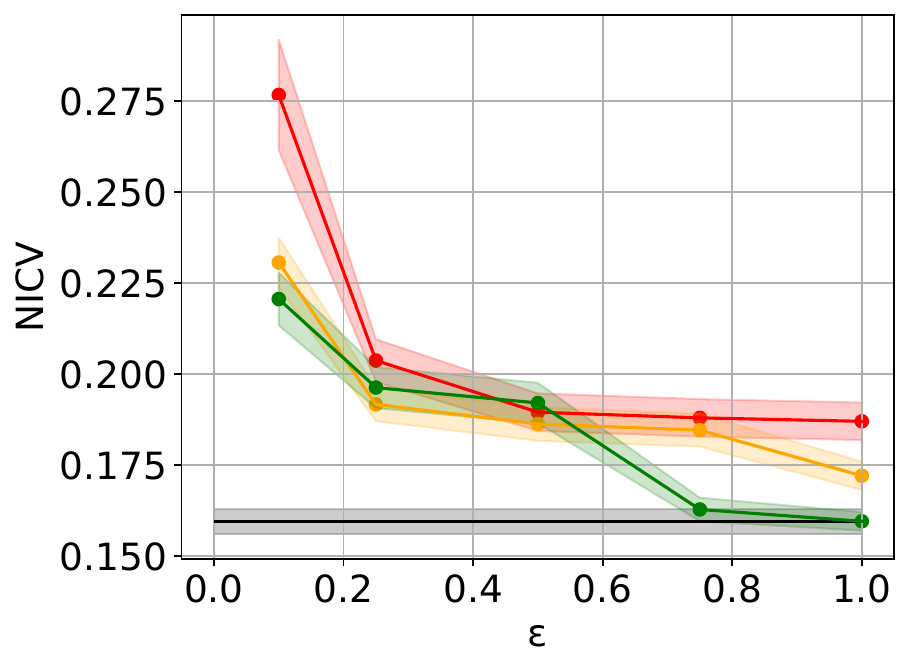}
        \caption{House}
        \label{fig:house}
    \end{subfigure}
    \begin{subfigure}{0.19\textwidth}
        \includegraphics[width=\textwidth]{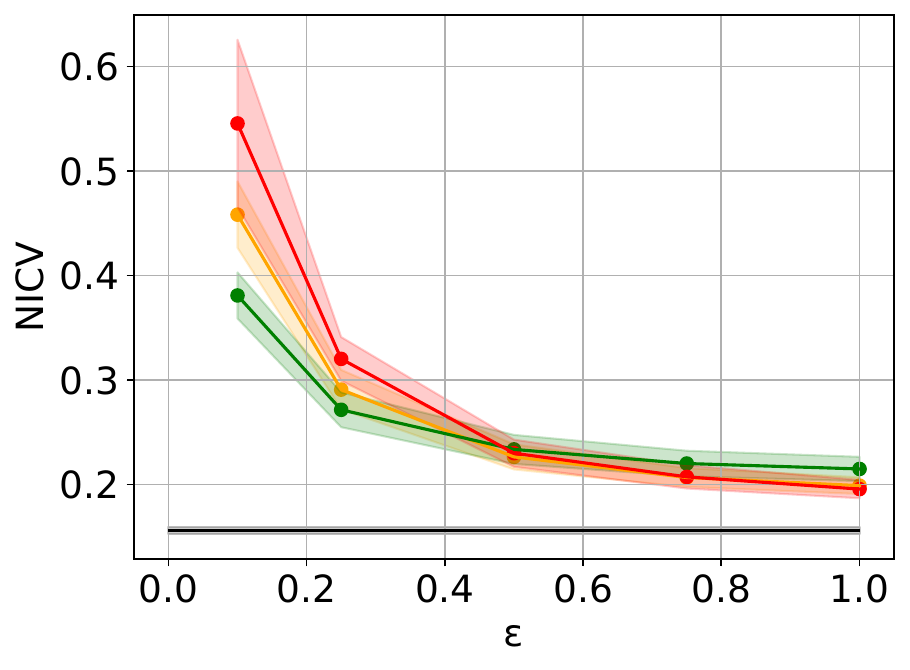}
        \caption{LSun}
        \label{fig:lsun}
    \end{subfigure}
    \begin{subfigure}{0.19\textwidth}
        \includegraphics[width=\textwidth]{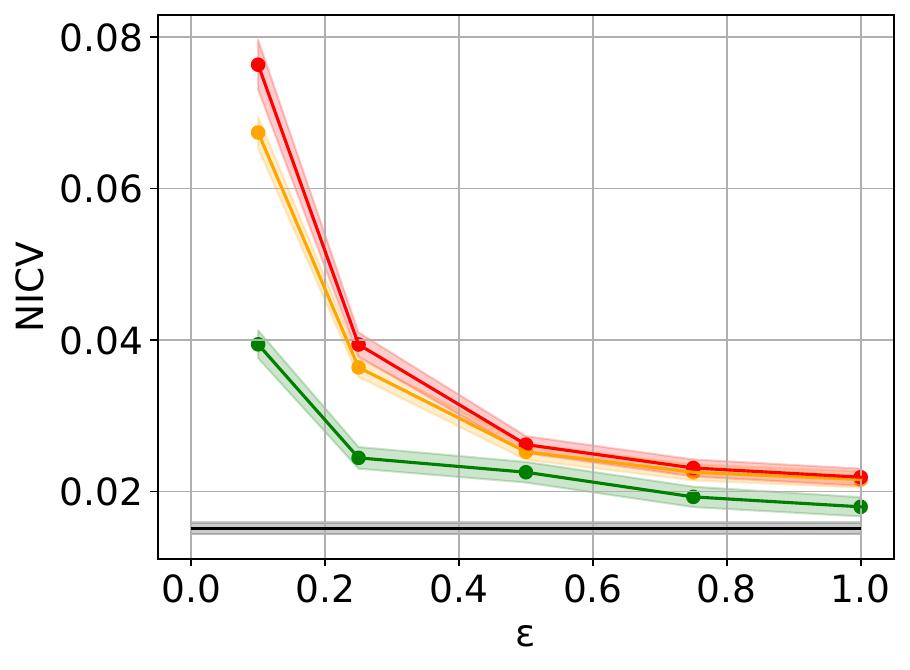}
        \caption{S1}
        \label{fig:s1}
    \end{subfigure}
    \begin{subfigure}{0.19\textwidth}
        \includegraphics[width=\textwidth]{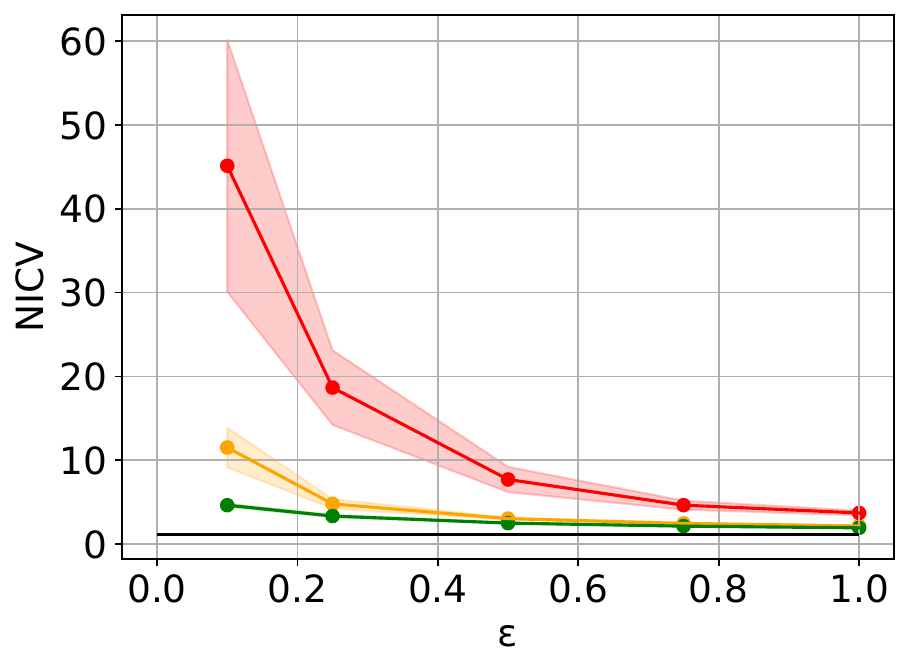}
        \caption{Wine}
        \label{fig:wine}
    \end{subfigure}
    \begin{subfigure}{0.19\textwidth}
        \includegraphics[width=\textwidth]{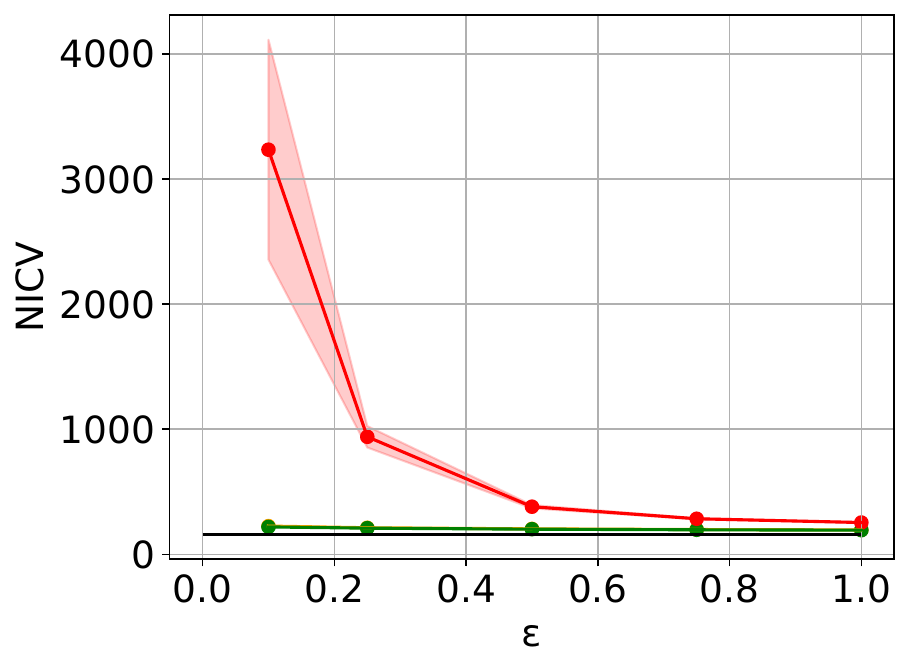}
        \caption{MNIST}
        \label{fig:mnist}
    \end{subfigure}
    \caption{Normalized Intra-cluster Variance (NICV) vs $\epsilon$ for real datasets.}\label{fig:util_compare}
\end{figure*}

\subsection{Experimental Setup}
\subsubsection{Implementation}
Our experimental evaluation was conducted on a Macbook Pro M2 Max (30-core CPU, 38-core GPU, 64GB RAM) using Open MPI~\cite{gabriel04:_open_mpi} for multiparty communication. Following Mohassel et al.~\cite{Mohassel2020}, we evaluate in the LAN setting with simulated network latency (0.25ms per send operation), noting that WAN runtimes can be derived using a linear cost model. Our experimental framework utilizes a two-client setup with balanced dataset partitioning. The MSA protocol (Figure~\ref{fig:MSA_protocol}) operates over the ring $\mathbf{R} = \mathcal{Z}_{2^{32}}$, employing fixed-point representation with precision factor $\Prec=16$. All experiments use randomly partitioned datasets, with results averaged over 100 runs and reported with 95\% confidence intervals where applicable.  For \ourProtocol, we set the radius constraint $\alpha = 0.8$ using the ``Step'' strategy (see Appendix~\ref{app:ablation} for detailed analysis and justification of these parameter choices). We set $\delta=\frac{1}{\DatasetSize \log{\DatasetSize}}$ and report the total $\epsilon$ over all iterations for all experiments. Our implementation is publicly available at \url{https://doi.org/10.5281/zenodo.15530617}.

\subsubsection{Baselines}
We evaluate \ourProtocol~against several baselines. For evaluating utility, we implement three protocols. First, \lloydsProtocol~\cite{Lloyd1982}, a non-private Lloyd's algorithm adapted to our federated framework with sphere packing initialization. Second, \suProtocol~\cite{Su2016}, which adapts Su et al.'s centralized differentially private \km clustering algorithm to our federated setting using our masked secure aggregation (MSA). Third, \gProtocol, a modification of \suProtocol~replacing Laplace noise with Gaussian noise and using composition and privacy budget analysis similar to our work. In Appendix~\ref{app:gaussian}, we provide the details of how the analysis in Section~\ref{sec:dp_params} differs for \gProtocol.  For computational and communication efficiency benchmarking, we compare against \mohProtocol~\cite{Mohassel2020}, established by Hedge et al.~\cite{Hegde2021} as the most efficient secure \km protocol. 

\subsubsection{Datasets}
Our evaluation employs both real and synthetic datasets. We use established datasets from the clustering datasets repository~\cite{ClusteringDatasets}, following Su et al.~\cite{Su2016} and Mohassel et al.~\cite{Mohassel2020} for comparability. For Birch2~\cite{Birchsets}, we take $25,000$ random samples from the $100,000$ sample dataset.

For evaluating scalability, following~\cite{Su2016}, we generate synthetic datasets (\textit{Synth}) using the \texttt{clusterGeneration} R package~\cite{clusterGen}, which enables control over inter-cluster separation (in $ [-1,1]$). The synthetic datasets contain $\DatasetSize = 10,000$ samples across $\NumClusters$ clusters, with cluster sizes following a $1:2:...:\NumClusters$ ratio. We incorporate a random number (in $[0, 100]$) of randomly sampled outliers and set the cluster separation degrees in $[0.16, 0.26]$, spanning partially overlapping to separated clusters. While Su et al.~\cite{Su2016} evaluate configurations up to $\NumClusters=10$ and $\Ddim=10$, we extend the evaluation to $\NumClusters=32$ and $\Ddim=512$ in powers of two, creating $45$ parameter combinations. Each combination generates three datasets with different random seeds. To assess scalability at higher cluster counts, we further extend our evaluation with \textit{Synth-K}, incorporating configurations up to $\NumClusters=128$ for $\Ddim=2$.

For benchmarking performance, we create synthetic datasets (\textit{TimeSynth}) with balanced cluster sizes ($C_{avg} = \frac{\DatasetSize}{\NumClusters}$), varying \DatasetSize{}~(10K-100K), \Ddim{}~(2-5), and \NumClusters{}~(2-5), following Mohassel et al.~\cite{Mohassel2020}.
Table~\ref{tab:dataset_summary} summarizes these datasets. Following standard practice in DP literature~\cite{Blum2005,Mcsherry2009, Dwork2011, Su2016}, all datasets are normalized to $[-1,1]$.

\begin{table}[h]
    \centering
    \caption{Summary of Datasets Used in Evaluation}
    \label{tab:dataset_summary}
    \begin{tabular}{|l|l|l|l|}
    \hline
    \textbf{Dataset} & \textbf{\DatasetSize{}} & \textbf{\Ddim{}} & \textbf{\NumClusters} \\
    \hline
    Iris~\cite{iris_53} & 150 & 4 & 3 \\
    LSun~\cite{ultsch2005lsun} & 400 & 2 & 3 \\
    S1~\cite{S1} & 5000 & 2 & 15 \\ 
    House~\cite{ClusteringDatasets} & 1837 & 3 & 3 \\
    Adult~\cite{misc_adult_2} & 48842 & 6 & 3 \\
    Wine~\cite{wine} & 178 & 13 & 3 \\
    Breast~\cite{breast} & 699 & 9 & 2 \\
    Yeast~\cite{yeast} & 1484 & 8 & 10 \\
    MNIST~\cite{lecun1998mnist} & 10000 & 784 & 10 \\
    Birch2~\cite{Birchsets} & 25000 & 2 & 100 \\
    G2~\cite{G2sets} & 2048 &  2-1024 &  2\\
    \hline
    Synth & 10K & 2-512 & 2-32 \\ 
    Synth-K & 10K & 2 & 2-128 \\ 
    TimeSynth & 10K, 100K & 2,5 & 2,5 \\
    
    \hline
    \end{tabular}
    \end{table}

\begin{figure*}[t]
    \centering
    \begin{subfigure}{\textwidth}
        \centering
        \includegraphics[width=0.5\textwidth]{figs/metrics/legend.pdf}
    \end{subfigure}
    \begin{subfigure}{0.49\textwidth}
        \centering
        \includegraphics[width=\columnwidth]{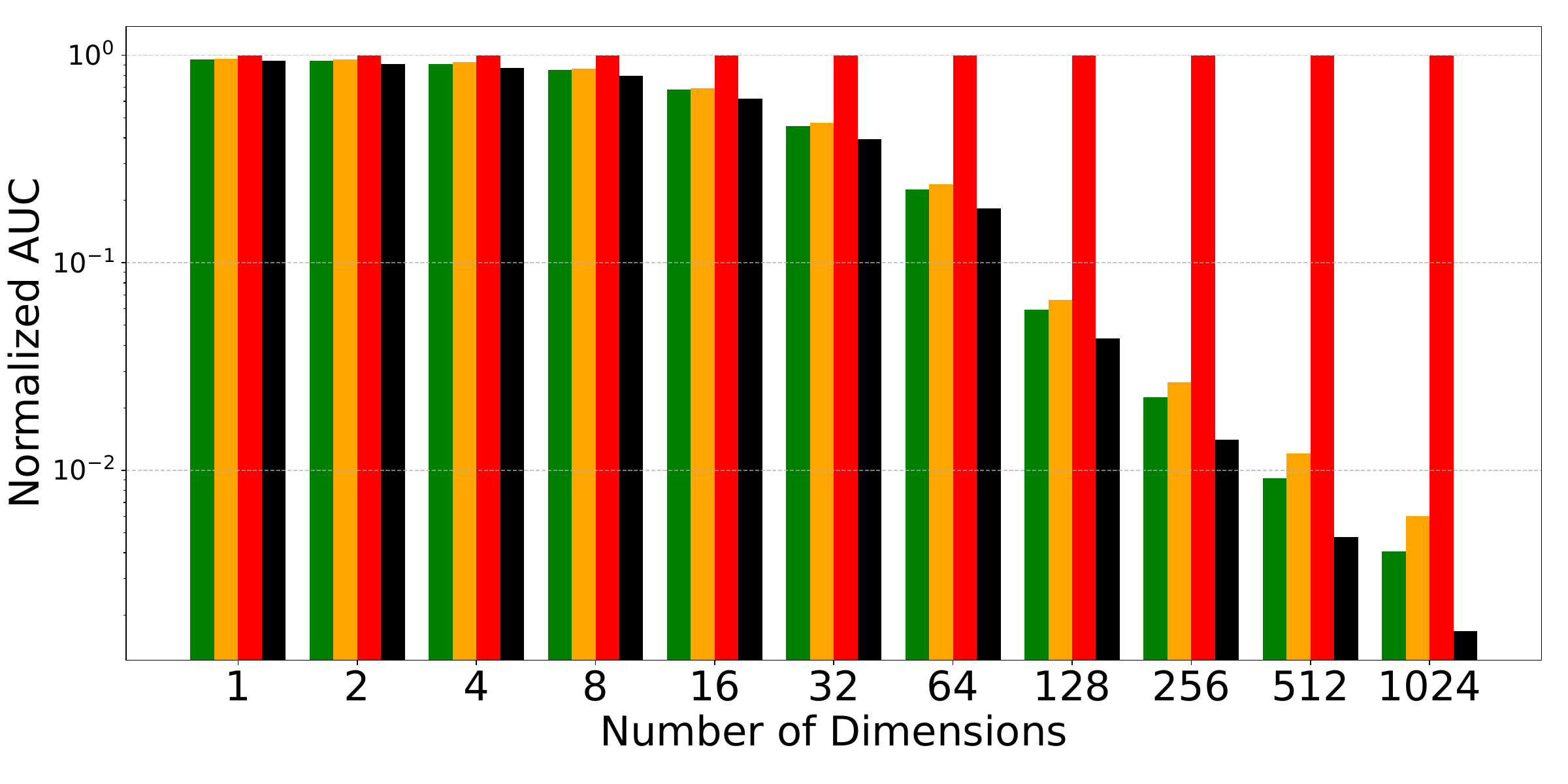}
        \caption{Varying the number of dimensions using G2 ($\NumClusters=2$, log scale).}
        \label{fig:g2_auc}
    \end{subfigure}
\begin{subfigure}{0.49\textwidth}
    \centering
    \includegraphics[width=\columnwidth]{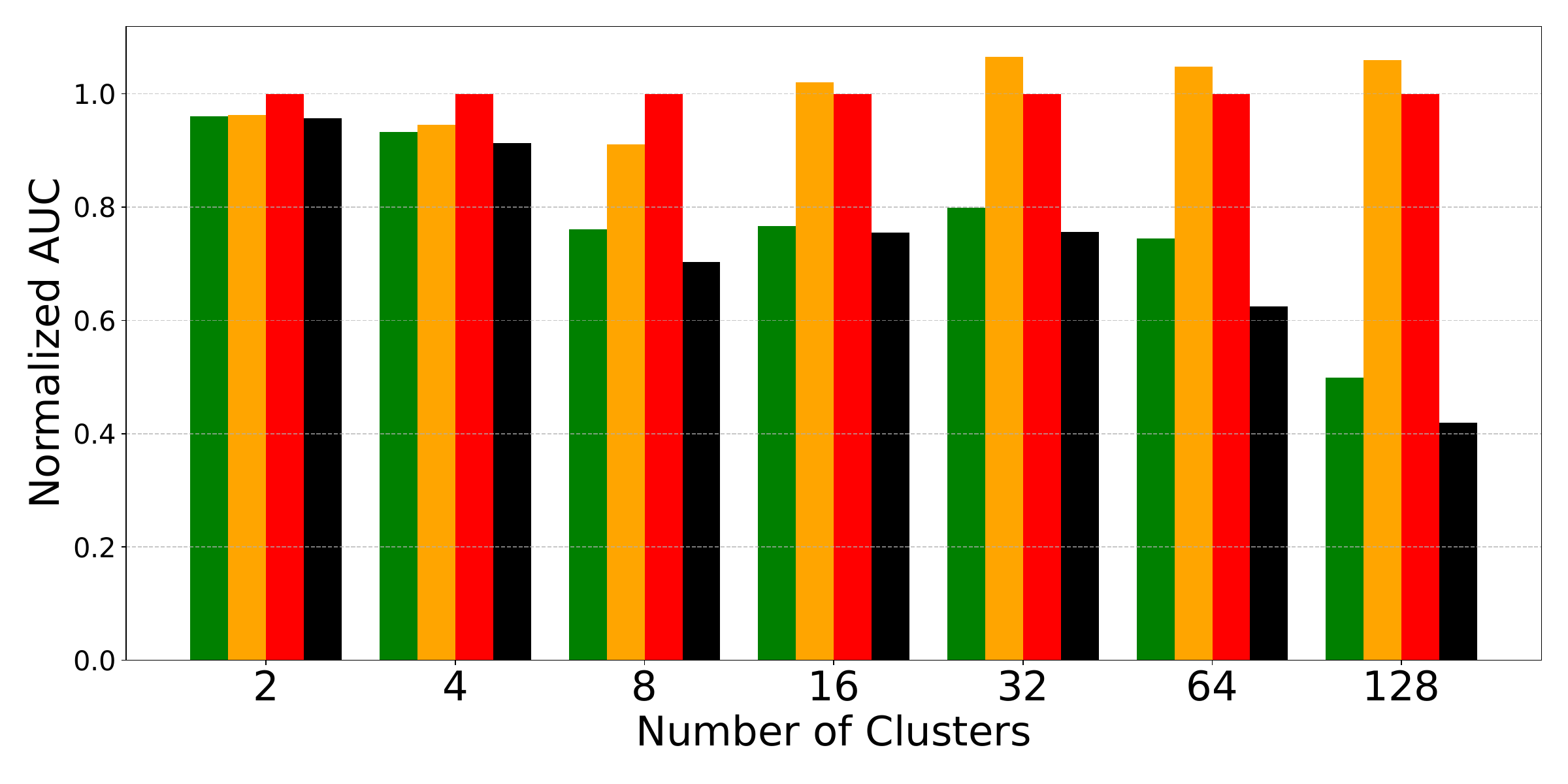}
    \caption{Varying the number of clusters using \textit{Synth-K} ($\Ddim=2$).}
    \label{fig:synth_auc}
\end{subfigure}
\caption{Scaling of AUC over number of dimensions and clusters. }
\label{fig:dimclust_scale}
\end{figure*}
\begin{figure*}[t]
    \centering
    \begin{subfigure}{0.33\textwidth}
        \includegraphics[width=\textwidth]{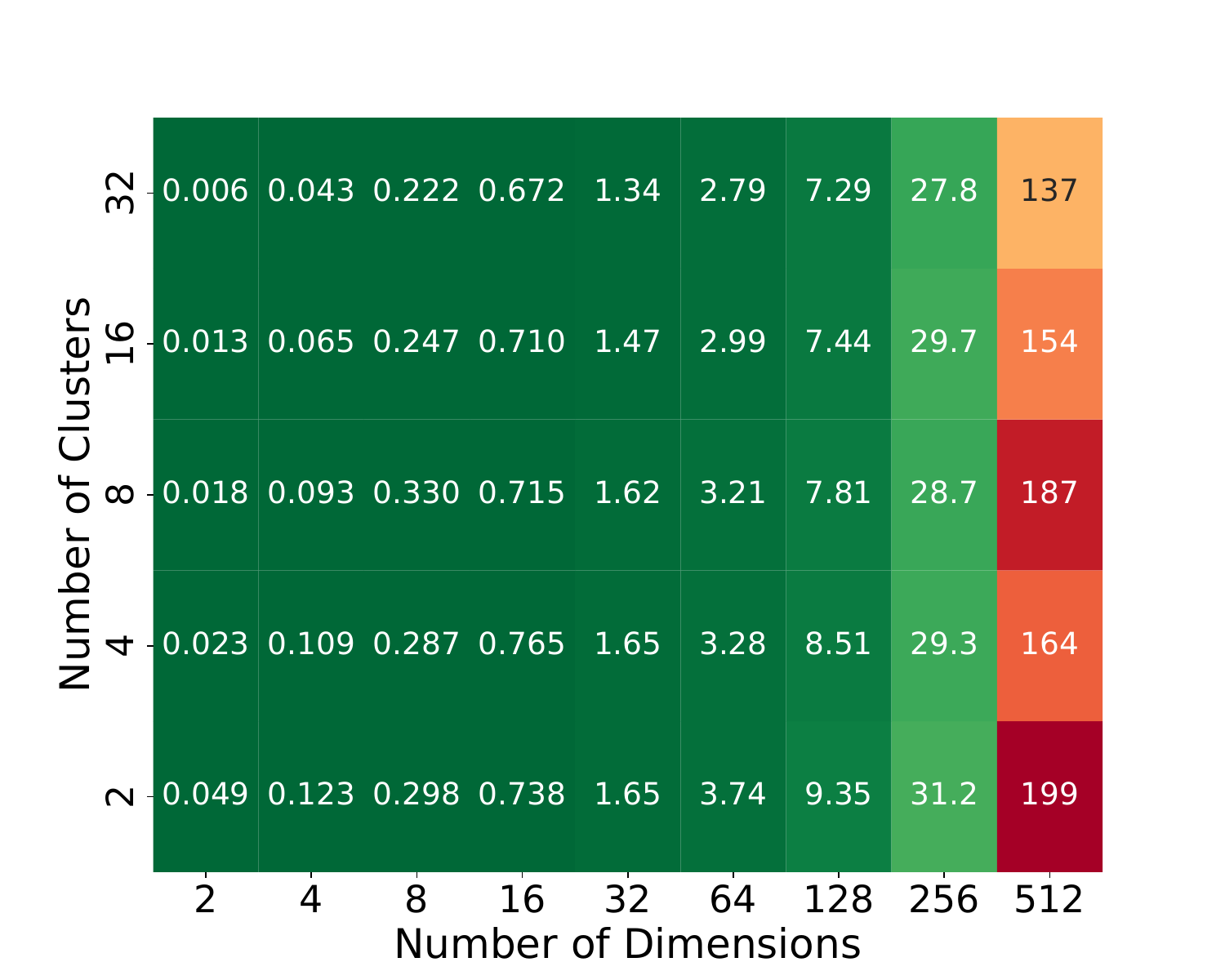}
        \caption{\suProtocol}
        \label{fig:su_heatmap}
    \end{subfigure}
    \begin{subfigure}{0.33\textwidth}
        \includegraphics[width=\textwidth]{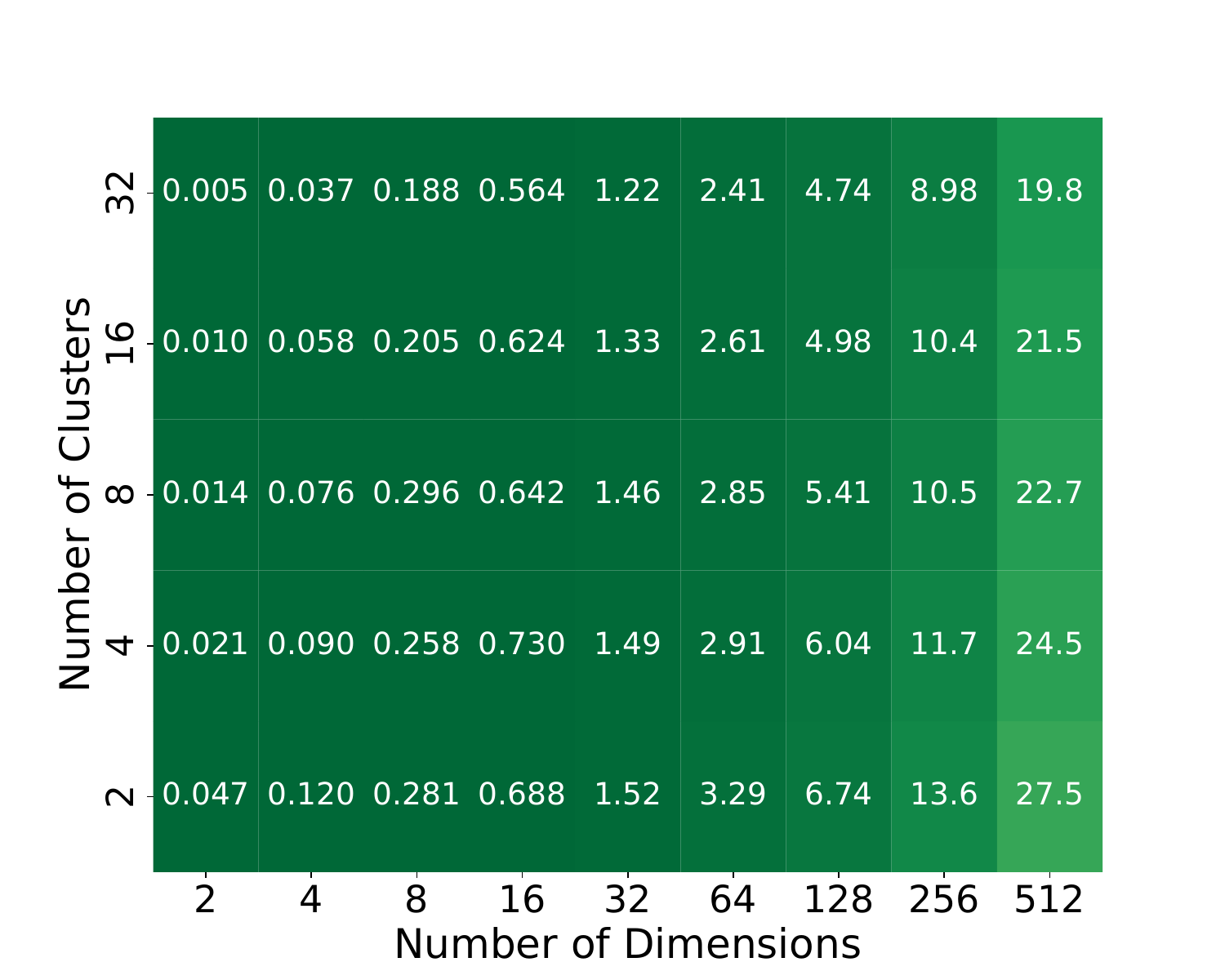}
        \caption{\ourProtocol}
        \label{fig:our_heatmap}
    \end{subfigure}
    \begin{subfigure}{0.33\textwidth}
        \includegraphics[width=\textwidth]{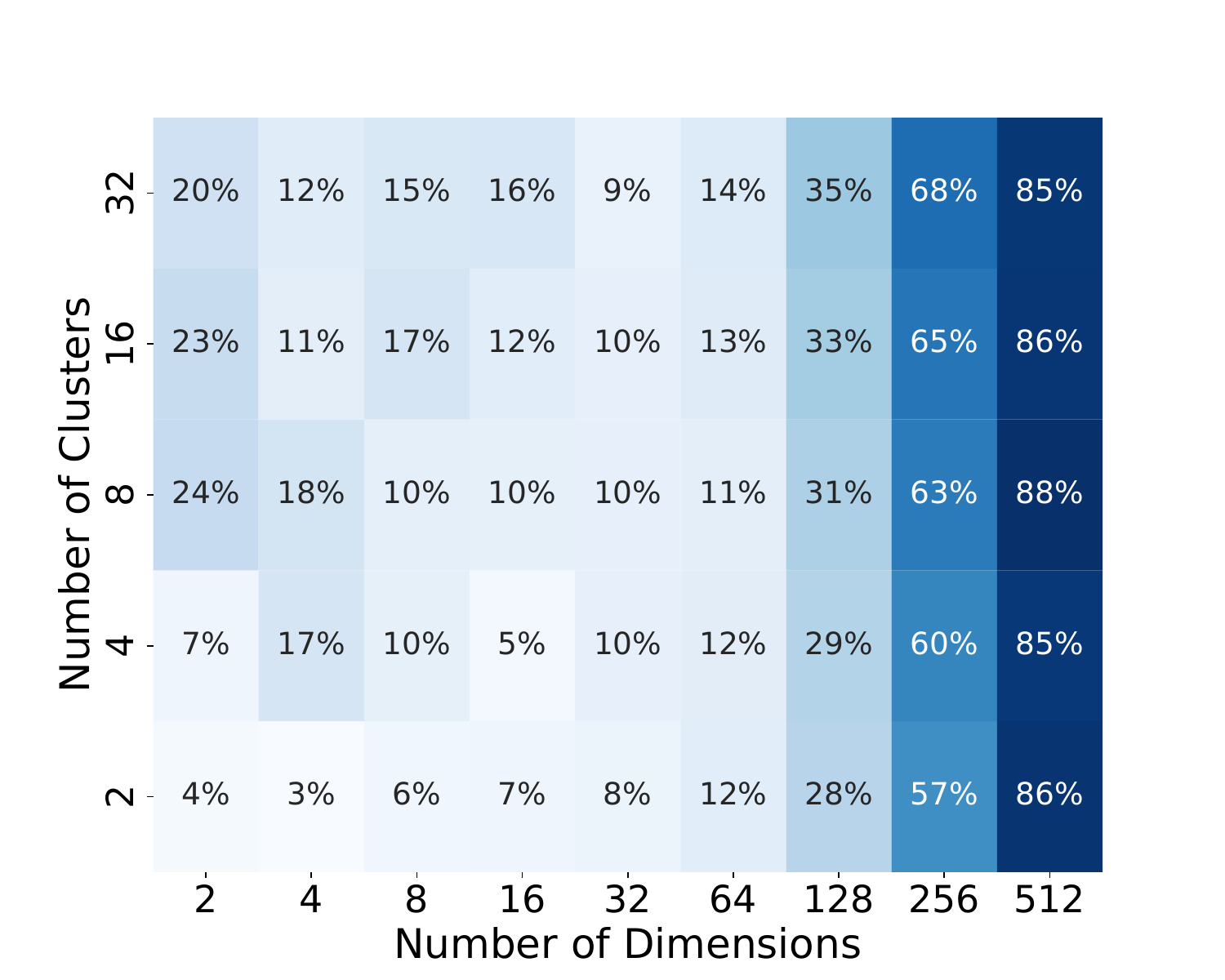}
        \caption{Improvement}
        \label{fig:improvement_heatmap}
    \end{subfigure}
    \caption{AUC comparison of \ourProtocol~vs. \suProtocol~on the \textit{Synth} datasets.}\label{fig:util_scale}
    \end{figure*}

\subsection{Utility Evaluation}\label{sec.utility_eval}

\subsubsection{Metrics}
We use the Normalized Intra-cluster Variance (NICV) from prior work~\cite{Su2016} as our primary utility metric. NICV normalizes the \km objective function (Eqn~\ref*{eqn:wcss_obj}) by dividing it by the dataset size:

\begin{equation*}
    NICV(\ClusterI{}) = \frac{1}{\DatasetSize} \cdot \sum_{\ClusterID=1}^{\NumClusters} \sum_{\datapointi{\DatasetIndex} \in \ClusterI{\ClusterID}} ||\datapointi{\DatasetIndex} - \CenterI{\ClusterID}||^2 
\end{equation*}

To facilitate systematic comparison across methods for the synthetic datasets, we summarize each method's performance over different privacy budgets using the area under the curve (AUC) of NICV values against $\epsilon \in [0.1, 0.25, 0.5, 0.75, 1.0]$, computed via the trapezoidal rule:

\begin{equation}
    AUC = \sum_{i=1}^{n-1} \frac{NICV_i + NICV_{i+1}}{2} \cdot (\epsilon_{i+1} - \epsilon_i)
\end{equation}

\ifthenelse{\boolean{extended}}{
We evaluate additional metrics in Appendix~\ref{sec:appendix_util_eval}.}{
Due to space constraints, we defer the evaluation of additional metrics to the extended version of this paper~\cite{full_version}.}

\subsubsection{Real World Datasets}
Figure~\ref{fig:util_compare} shows the NICV against various privacy budgets on ten real-world datasets to answer \ref{q:util_compare}.
The shading shows the 95\% confidence interval of the mean over the 100 experiments.
We observe that \ourProtocol{}~and \gProtocol{}~consistently outperform \suProtocol{}~across all high-dimensional datasets (e.g., Wine, Breast, Yeast, MNIST). 
This aligns with our theoretical expectations as the Gaussian mechanism composes much more favourably in high-dimensional spaces than the Laplace mechanism. 
Additionally, \ourProtocol{}~outperforms \gProtocol{}~in all cases, especially in datasets with low dimensions and a high number of clusters (e.g., Birch2, S1), where \gProtocol{}~does not outperform \suProtocol{}.
This highlights the effect of imposing a radius constraint on the protocol, which, as discussed in Section~\ref{sec:constrained_k_means}, is most restrictive in low-dimensional, high-cluster settings, further reducing the noise added to the centroids.
While all differentially private methods eventually approach the non-private baseline as privacy budget increases, \ourProtocol~achieves this convergence at substantially lower privacy budgets, establishing it as the superior approach for DP \km clustering across all datasets.

\subsubsection{Synthetic Datasets}
To answer \ref{q:util_scale}, we follow an approach similar to Su et al.~\cite{Su2016} and evaluate \ourProtocol{}~and \suProtocol{}~on our synthetic datasets varying both \NumClusters{}~and \Ddim.
Figure~\ref{fig:util_scale} shows heatmaps of the AUC values for \suProtocol{}~and \ourProtocol{}~on the synthetic datasets.
To clarify the difference between the protocols, Figure \ref{fig:improvement_heatmap} shows the percentage of reduction in AUC of \ourProtocol~over that of \suProtocol.
From the heatmaps, we observe that \ourProtocol{}~outperforms \suProtocol{}~across all datasets, with the improvement being more pronounced in higher dimensions and in higher \NumClusters{}.

To investigate higher limits, we further extend the evaluation to the G2~\cite{G2sets} datasets with $\Ddim$ up to 1024, and \textit{Synth-K} datasets with $\NumClusters$ up to 128.
Figures~\ref{fig:g2_auc} and~\ref{fig:synth_auc} present a comparative analysis of AUC performance across \lloydsProtocol, \ourProtocol{}, \gProtocol{}, and \suProtocol{} protocols, evaluated on G2 and \textit{Synth-K} datasets. The AUC values are normalized relative to \suProtocol{} as the baseline. Figure~\ref{fig:g2_auc} examines dimensionality scaling by varying $\Ddim$ up to 1024 while maintaining $\NumClusters=2$, whereas Figure~\ref{fig:synth_auc} shows cluster scaling by varying $\NumClusters$ up to 128 while fixing $\Ddim=2$.
In Figure~\ref{fig:g2_auc}, we observe that for low dimensions, all protocols perform very closely to the baseline (\lloydsProtocol), but as the dimension increases, the task becomes much more challenging for the private protocols, with \ourProtocol~outperforming all others.
In Figure~\ref{fig:synth_auc}, we observe that \ourProtocol~outperforms all other protocols across all cluster counts, with the improvement being more pronounced at higher cluster counts. This is to be expected as the radius constraint becomes more restrictive with higher cluster counts, leading to a more significant reduction in noise added to the centroids. 

Since Su et al. found \suProtocol{}~outperformed all other approaches for $\Ddim{}>3$~\cite{Su2016}, and we outperform \suProtocol{}, we claim that \ourProtocol{}~is state-of-the-art in DP \km for $\Ddim{}>3$.
\ourProtocol{}~also fixes the scalability issue that Su et al.'s work~\cite{Su2016} observed in high \NumClusters{} and enables us to scale to much larger dimensions.

\begin{table*}
    \centering
        \begin{tabular}{llllllllll}
\toprule
Dataset & \multicolumn{3}{c}{Parameters} & \multicolumn{3}{c}{Runtime (ms)} & \multicolumn{3}{c}{Comm (bytes)}\\
\cmidrule(lr){2-4}\cmidrule(lr){5-7}\cmidrule(lr){8-10}
& \DatasetSize{} & \NumClusters{} &\Ddim{}& \ourProtocol{} & \mohProtocol~\cite{Mohassel2020} & Speedup & \ourProtocol{} & \mohProtocol~\cite{Mohassel2020} & Reduction\\
\midrule
TimeSynth & 10K & 2 & 2 &   $3.03 \pm 0.02$ &   10500 &     $3465\times$ &  192& 2.557e8 &$1.33e6\times$ \\
& & 2 & 5 &                 $3.24 \pm 0.01$ &   - &         - &             384& - & - \\
& & 5 & 2 &                 $4.32 \pm 0.02$ &   34050 &     $7882\times$ &  480& 9.746e8& $2.03e6\times$ \\
& & 5 & 5 &                 $4.95 \pm 0.04$ &   - &         - &             960& - & - \\
TimeSynth & 100K & 2 & 2 &  $12.57 \pm 0.02$ &  105120 &    $8363\times$ &  192& 2.467e9 & $1.28e7\times$ \\
& & 2 & 5 &                 $13.17 \pm 0.03$ &  - &         - &             384& - & - \\
& & 5 & 2 &                 $22.05 \pm 0.03$  & 347250 &    $15748\times$ & 480& 9.535e9 & $1.99e7\times$ \\
& & 5 & 5 &                 $22.56 \pm 0.02$ &  - &         - &             960& - & - \\
\midrule
LSun & 400 & 3 & 2 &        $2.6 \pm 0.04$ &    1481 &      $570\times$ &   288& - & - \\
S1 & 5K & 15 & 2 &          $5.81 \pm 0.08$  &  49087 &     $8449\times$ &  1440& - & - \\
\bottomrule
\end{tabular}

        \caption{Overhead comparison per iteration against \mohProtocol~\cite{Mohassel2020} for two clients ($100$ runs with mean and $95\%$ confidence reported for \ourProtocol)}\label{tab:overhead}
\end{table*}
\subsection{Runtime Evaluation}\label{sec.runtime_eval}
To answer questions \ref{q:time_compare} and \ref{q:time_scale}, we compare \ourProtocol~with the \mohProtocol~protocol of Mohassel et al.~\cite{Mohassel2020} on a variety of common datasets in terms of runtime and communication size (per iteration).
Table \ref*{tab:overhead} shows the summary of the evaluation with two clients (the setting considered in Mohassel et al.'s work~\cite{Mohassel2020}).
Values for \mohProtocol~are taken as reported in the paper by Mohassel et al.~\cite{Mohassel2020}, noting that the setup is similar to the one used in our evaluation, and the gap in performance is more than what could be accounted for by different setups.
In terms of runtime, \mohProtocol~executes in the order of minutes, while \ourProtocol~executes in milliseconds, offering five orders of magnitude speedup.
In terms of communication size, \mohProtocol~requires communicating gigabytes of data per iteration, while \ourProtocol~needs a fraction of a kilobyte, offering up to seven orders of magnitude reduction in size.
In terms of communication rounds, \mohProtocol~requires $\Theta (\lceil \log{\NumClusters{}}\rceil)$ communication rounds per iteration~\cite{Hegde2021}, while \ourProtocol~only requires one.
Since \mohProtocol~was found to be the state-of-the-art in secure federated \km~\cite{Hegde2021}, we conclude that we advance the state-of-the-art while offering output privacy (which would only further slow down Mohassel et al.~\cite{Mohassel2020}).

\section{Conclusion}
In this work, we designed \ourProtocol, a new private protocol for federated \km.
We have shown that \ourProtocol~is secure in the computational model of DP and analyzed its utility.
Compared to state-of-the-art solutions in the central model of DP, \ourProtocol~results in higher utility across a wide range of real datasets and scales effectively to larger dimensions and number of clusters.
\ourProtocol~also achieves five orders of magnitude faster runtime than the state-of-the-art in secure federated \km~across a variety of problem sizes.
In summary, we provide an efficient, private, and accurate solution to the horizontally federated \km problem.

\section*{Acknowledgments}

We gratefully acknowledge the support of NSERC for grants RGPIN-2023-03244, IRC-537591, the Government of Ontario and the Royal Bank of Canada for funding this research.

\section*{Open Science}

We make all source code and datasets used in our paper available here: \url{https://doi.org/10.5281/zenodo.15530617}. 
This repository includes all relevant code and supporting scripts necessary to reproduce our experiments and results.

\section*{Ethics Considerations}
Our work develops a new protocol for privately clustering data.
Our primary stakeholders are the data scientists who will deploy our protocol.
The secondary stakeholders are the subjects of the analyses whose data is aggregated and analyzed in this protocol.
We improve the accuracy, privacy, and runtime of existing protocols in the private clustering domain.
Simultaneous improvements in all three categories have positive ethical implications for both stakeholders.
Improved privacy protection (input and output privacy) reduces the risks for both stakeholders.
Our accuracy improvements make the output more useful to the primary stakeholders.
Finally, our runtime improvements reduce the computational cost and, thus, the environmental impact of conducting the analysis.
All of these improvements incentivize the use of private protocols over non-private alternatives.
While our improvements significantly reduce the risks compared to related work, care must be taken to appropriately communicate the inherent risks of deploying any protocol that satisfies a similar security model to both stakeholders. 

\bibliographystyle{plain}
\bibliography{refs}

\appendix
\section{Proofs}
\subsection{Proof of Theorem~\ref{thm:maxdist_sens}}\label{app:proof_maxdist_sens}
\RadiusSensitivity*
\begin{proof}
    w.l.o.g assume that the datasets differ by a single point $\datapointi{}'$ so that $\Dataset' = \Dataset \cup \{\datapointi{}'\}$.
    If $\datapointi{}'$ is not within \MXD~of $\NoiseCenterI{\ClusterID}^{(\IT-1)}$, it will not be assigned to $\ClusterI{\ClusterID}$ and thus $\Diff{\ClusterID}(D) = \Diff{\ClusterID}(D')$.
    Therefore, we only need to consider the case where $\datapointi{}'$ is at most \MXD~from $\NoiseCenterI{\ClusterID}^{(\IT-1)}$.
    By definition, we have:
    \begin{eqnarray*}
        \Delta^{\Diff{}} &=& \max\limits_{\Dataset, \Dataset' \in \mathcal{D}} ||\Sm{\ClusterID}^{(\IT)} - \Cnt{\ClusterID}^{(\IT)}\NoiseCenterI{\ClusterID}^{(\IT-1)} - (\Sm{\ClusterID}'^{(\IT)} - \Cnt'{\ClusterID}^{(\IT)}\NoiseCenterI{\ClusterID}^{(\IT-1)})||_2 \\
        &=& \max\limits_{\Dataset, \Dataset' \in \mathcal{D}} ||\Sm{\ClusterID}^{(\IT)} - \Cnt{\ClusterID}^{(\IT)} \NoiseCenterI{\ClusterID}^{(\IT-1)} - \Sm{\ClusterID}^{(\IT)} - \datapointi{}' + (\Cnt{\ClusterID}^{(\IT)} +1) \NoiseCenterI{\ClusterID}^{(\IT-1)}||_2\\
        &=& \max\limits_{\Dataset, \Dataset' \in \mathcal{D}} || \NoiseCenterI{\ClusterID}^{(\IT-1)} - \datapointi{}'||_2\\
        &\leq& \MXD
    \end{eqnarray*}
    where the last line follows by (Eqn \ref{eq:maxdist_constraints}).
\end{proof}
\subsection{Proof of Theorem~\ref{thm:sec_proof}}\label{app:proof_sec_proof}
We first prove the following helpful lemma.
\begin{lemma}\label{lem:GDP_analysis}
    Our noise mechanism (defined in Eqn~\ref{eq:sum_noise}) of $f(D) + \DPNoise$ where $\DPNoise \sim \mathcal{N}(0,\sigma^2T(\Delta^{(f)})^2)$ applied over $T$ adaptive iterations of $f$ is $(\frac{1}{\sigma})$-GDP, where $\Delta^{(f)}$ is the sensitivity of $f$.
\end{lemma}

\begin{proof}
    After choosing $\sigma$, by Theorem~\ref{thm:g_mech} from Dong et al.~\cite[Theorem 1]{dong22_gdp}, we have that each application of a Gaussian mechanism with noise multiplier $\sigma$ is $\frac{\Delta^{(f)}}{\sigma}$-GDP.
    We apply the Gaussian mechanism adaptively over $T$ iterations. Which by Theorem~\ref{thm:gdp_comp} from Dong et al.~\cite[Corollary 2]{dong22_gdp}, gives us $\frac{\Delta^{(f)}\sqrt{T}}{\sigma}$-GDP.
    Thus, by multiplying $\sigma$ by $\sqrt{T}\Delta^{(f)}$, we get $\frac{1}{\sigma}$-GDP.
\end{proof}
We now restate and prove Theorem~\ref{thm:sec_proof}.
\SecProof*
\begin{proof}
    To prove the algorithm satisfies Definition~\ref{def:sec_model}, we need to consider the view from each party.
    We begin with the view of the server.
    First, let us assume that each client samples the random mask $\MM{\PartyIndex}$ from a random oracle over $\mathbf{R}$.
    Under this assumption, the $Enc(\cdot)$ function satisfies information-theoretic security as it is a one-time pad using $\MM{\PartyIndex}$ as the pad.
    In practice, we implement the sampling of $\MM{\PartyIndex}$ using a PRNG and thus reduce from information-theoretic to computational security with a negligible term $negl(\lambda)$.
    Therefore, the view of the server satisfies $\epsilon(\lambda)$-IND-CDP-MPC with $\epsilon=0$.
    
    Regarding the clients, a client's view consists of their own dataset $D_A$, and the output from each iteration of the protocol.
    We note that the initialization is data-independent and thus is indistinguishable between neighbouring datasets.
    Each iteration consists of the assignment and updating of the cluster centroids.
    The assignment step uses the published centroids from the previous iteration (post-processing) to divide the dataset into clusters. We can then apply parallel composition over each of the clusters.
    Thus, we can focus on the privacy cost of a single cluster for the remainder of the proof.
    
    First, we set $\sigma$ using Theorem~\ref{thm:GDPtoDP} from Dong et al.~\cite[Corollary 1]{dong22_gdp}.
    We choose the minimum $\sigma$ such that $(\epsilon,\delta)$-DP iff $\frac{1}{\sigma}$-GDP by Theorem~\ref{thm:GDPtoDP}.
    This is equivalent to finding the noise multiplier for a sensitivity one, single application, of the Gaussian mechanism with noise multiplier $\sigma$.
    To solve this minimization, we use Algorithm 1 from Balle and Wang~\cite{balle18_gaussian}.

    The updating of the cluster centroid applies the Gaussian mechanism twice, once to the relative sums and once to the counts.
    We split the overall $\sigma$ into $\sigma^\Diff{}$ and $\sigma^\Cnt{}$ following Eqn~\ref{eq:dp_params}.
    We begin with the analysis of the relative sum.
    In Theorem~\ref{thm:maxdist_sens}, we show that the sensitivity of the relative sum is $\MXD$.
    Applying Lemma~\ref{lem:GDP_analysis} we get that the relative sum over $T$ iterations is $\frac{1}{\sigma^\Diff{}}$-GDP.
    The sensitivity of the count is $1$, and thus the count is $\frac{1}{\sigma^\Cnt{}}$-GDP by similar analysis.
    Finally, applying Theorem~\ref{thm:gdp_comp}. We get that 
    \begin{equation}
        \sqrt{\left(\frac{1}{\sigma^\Diff{}}\right)^2 + \left(\frac{1}{\sigma^\Cnt{}}\right)^2} = \frac{1}{\sigma}
    \end{equation}
    and thus the we get $\frac{1}{\sigma}$-GDP over the entire protocol which implies $(\epsilon, \delta)$-DP (because of how we chose $\sigma$).
    Applying parallel composition over all clusters, the client's view is $(\epsilon,\delta)$-IND-CDP-MPC, with no computational assumption (as we use the information-theoretic DP properties of differential privacy). 
    Taking the worst case over the clients and the server, the result follows.
    \end{proof}
\subsection{Proof of Quantized DP Noise}\label{app:proof_quant}
We argue why quantizing the DP noise is acceptable in Figure~\ref{fig:MSA_protocol}. In essence, because the sensitive data is already quantized, adding quantized noise to it is equivalent to adding the noise first and then performing the quantization, which aligns with the post-processing lemma in differential privacy.
    
This can be intuitively understood by noting that:
    \begin{itemize}
        \item The data before noise addition is already quantized.
        \item The sum of a quantized value and a non-quantized value will only have information of the non-quantized variable in the less significant bits (i.e., those lost to quantization).
        \item Hence, quantizing the sum is equivalent to dropping this lower bit information, which is similar to quantizing the second variable prior to addition.
\end{itemize}
We show this formally in the following Theorem.
\begin{theorem}[Quantization and Differential Privacy]
    Quantizing Laplace noise at the same level of quantization as that used in the Masked Secure Aggregation (MSA) protocol does not violate the privacy guarantees offered by differential privacy.
    \end{theorem}
    \begin{proof}
    We begin by observing that any value, say $\VV{}$, can be split into two parts upon quantization: the integral part, $\Bar{\VV{}}$, and the fractional part, $\hat{\VV{}}$, such that $\Bar{\VV{}} \times \SF$ is the (rounded) integral part of $\VV{} \times \SF$ (i.e. it is $\round{\VV{} \times \SF}$), and $\hat{\VV{}}$ denotes the fractional part that gets lost due to quantization. This notation can similarly be applied to the noise variable, $\eta$. 
    
    In our context, the sensitive data we wish to protect is $\Bar{\VV{}}$ since the entire protocol operates in a quantized environment, i.e., $\VV{}$ is never transmitted. We aim to demonstrate the following equivalence:
    \[  \tilde{\VV{}}+\tilde{\eta} = \round{\VV{} \times \SF}+\round{\eta \times \SF} =  \round{(\Bar{\VV{}}+\eta) \times \SF} \]
    
    The right-hand side represents quantization applied as post-processing to the sum of the sensitive data and the noise, which adheres to the rules of differential privacy. 
    
    After expanding the right-hand side, we obtain:
    \[\round{(\Bar{\VV{}}+\eta) \times \SF}=\round{\Bar{\VV{}}\times \SF+ \Bar{\eta} \times \SF+ \hat{\eta} \times \SF}\]
    
    We can take out $\Bar{\VV{}}\times \SF$ and $\Bar{\eta{}}\times \SF$ from the rounding operator because they are exact integers (indeed $\Bar{\VV{}}\times \SF$ is $\round{\VV{} \times \SF}$). We then obtain:
    \[\round{(\Bar{\VV{}}+\eta) \times \SF}=\round{\VV{} \times \SF}+\round{\eta{} \times \SF}+\round{\hat{\eta} \times \SF}\]
    
    However, $\hat{\eta} \times \SF$ will be lost due to quantization, hence:
    \[\round{(\Bar{\VV{}}+\eta) \times \SF}=\round{\VV{} \times \SF}+\round{\eta{} \times \SF}=\tilde{\VV{}}+\tilde{\eta}\]
    \end{proof}

\section{Ablation of Parameters}\label{app:ablation}

\begin{figure}
    \centering
    \includegraphics[width=\columnwidth]{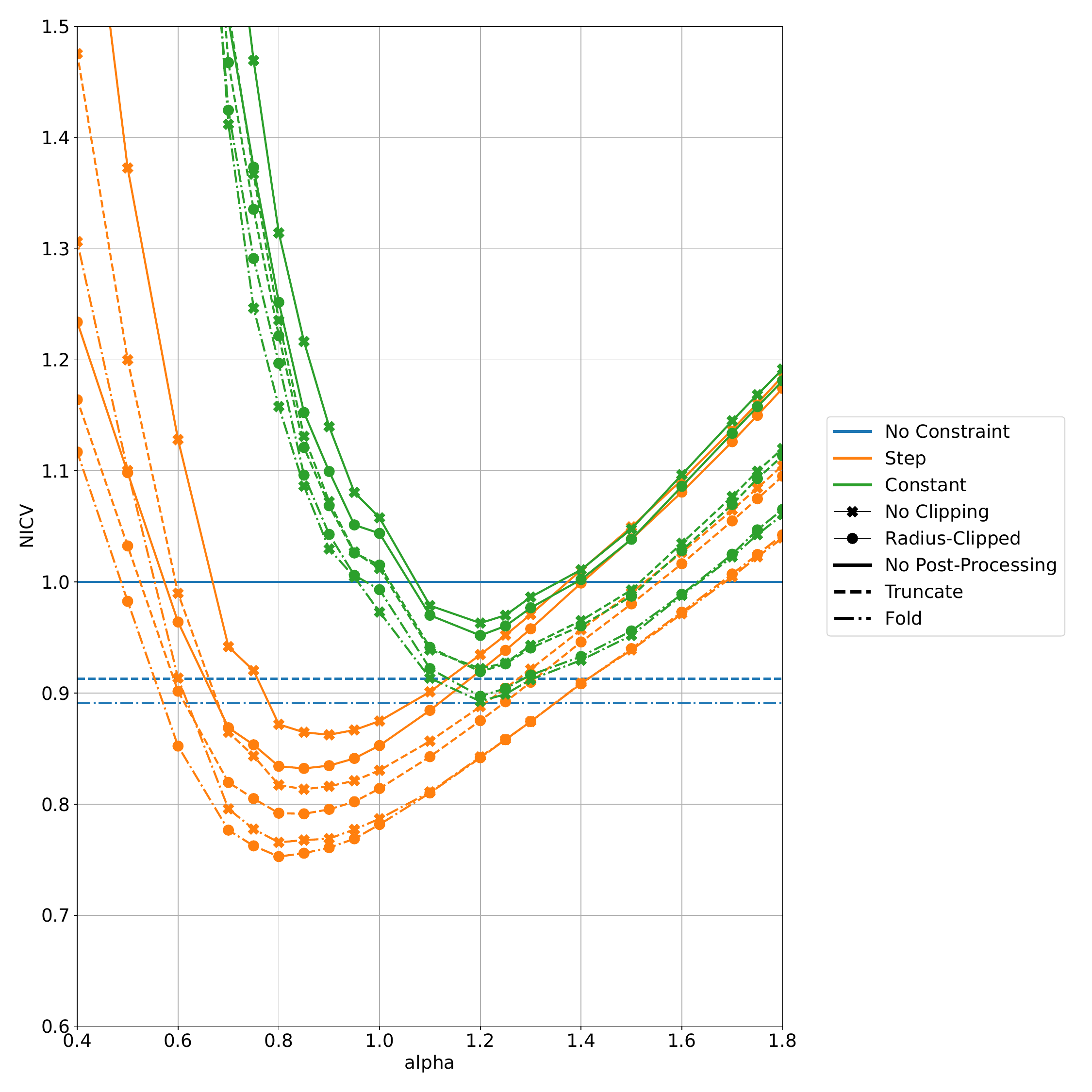}
    \caption{Ablation of $\alpha$ and post-processing strategy on NICV for $\epsilon=0.1$.}
    \label{fig:ablation}
\end{figure}
In this section, we empirically justify our choice of hyperparameter $\alpha$ and the post-processing strategy.
We choose these parameters to minimize the Normalized Intra-Cluster Variance (NICV) metric over a set of synthetic datasets.
The synthetic datasets were generated using the \texttt{clusterGeneration} package~\cite{clusterGen} varying the number of dimensions ($\Ddim$) and clusters ($\NumClusters$) as $(\NumClusters, \Ddim) \in \{2, 4, 8, 16\}^2$, resulting in 16 variations. 
For each variation, we generate three different datasets with the degrees of separation $\{0.25, 0.5, 0.75\}$. 
The average cluster size is fixed to $C_{avg} = \frac{2048}{\NumClusters}$ and the cluster sizes are randomly sampled in the range $[0.70 \cdot C_{avg}, 1.30 \cdot C_{avg}]$, leading to a total dataset size in the range $[1433, 2662]$.

All the ablation experiments were done on $\epsilon=0.1$.
The results are shown in Figure~\ref{fig:ablation}, where we plot the average NICV over all datasets for different values of $\alpha$ and post-processing strategies.
The average is then divided by the NICV of the naive baseline of no post-processing and no radius constraint.
We observe that the ``folding'' post-processing strategy consistently outperforms the ``truncation" and ``none'' strategies.
We also observe that the ``Step'' strategy outperforms the ``Constant'' strategy and provides most performance improvement for $\alpha \approx 0.8$.
We also note that ``radius clipping'' provides a consistent improvement.
Based on these results, we choose $\alpha=0.8$ and the ``Step'' constraint strategy with a post-processing strategy of ``folding" after ``radius clipping'' is applied.

\subsection{Error Analysis of Simple Gaussian Mechanism}\label{app:gaussian}

We follow a similar analysis to Section~\ref{sec:dp_params} to derive the parameters used for \gProtocol.
Since we do not modify the algorithm (only change the noise distribution), we can use Su et al.'s MSE analysis unchanged:
\begin{equation}
    \mse{\NoiseCenterI{}^{(\IT)}} \approx \frac{\NumClusters^3}{\DatasetSize^2}\left(\Var{\DPNoise^{\Sm{}}} + 4 \rho^2\Var{\DPNoise^{\Cnt{}}}\right)
\end{equation}

Substituting the variance following our noise mechanism (with the domain-based sensitivity) gives:
\begin{equation}
    \mse{\NoiseCenterI{}^{(\IT)}} \approx \frac{\NumClusters^3}{\DatasetSize^2}\left(\Ddim\NumIter(\sigma^{\Sm{}})^2 + 4 \rho^2\NumIter(\sigma^{\Cnt{}})^2\right).
\end{equation}
Minimizing this using Lagrange multipliers such that: 
\begin{equation}
    \sqrt{\left(\frac{1}{\sigma^\Diff{}}\right)^2 + \left(\frac{1}{\sigma^\Cnt{}}\right)^2} = \frac{1}{\sigma}
\end{equation}
gives the following ratio of noise multipliers:
\begin{equation}
    \sigma^\Cnt{} = \sqrt{\frac{\sqrt{\Ddim}}{2\rho}}\sigma^\Sm{}
\end{equation}
Substituting this back into the MSE equation gives:
\begin{equation}
    \mse{\NoiseCenterI{}^{(\IT)}} \approx \frac{\NumClusters^3 \Ddim \NumIter \sigma^2 (2\rho + \sqrt{\Ddim})^2}{\DatasetSize^2}
\end{equation}
Setting the per iteration MSE to be less than $0.004$ and rearranging for $\NumIter$ gives:
\begin{equation}
    \NumIter \leq \frac{\DatasetSize^2(0.004)}{\NumClusters^3 \Ddim \NumIter \sigma^2 (2\rho + \sqrt{\Ddim})^2}
\end{equation}
which we also truncate to be in $[2,7]$.

\section{Extensions}\label{sec:extensions}
\subsection{Serverless MSA}
Masked Secure Aggregation (MSA) is at the core of our protocol and requires a semi-honest server to perform the aggregation.
An alternative is for all or a subset of the clients to take the role of a computation node in an MPC protocol.
In this case, the centroids would be secret shared and aggregated in MPC, then DP noise would need to be added before release every iteration.
The noise addition can be implemented by each computation node sampling noise locally and adding it to their shares.
Depending on how many nodes are assumed to be honest, this will necessarily add more noise than \ourProtocol~to ensure privacy in the presence of one or more colluding nodes.
Alternatively, the computation nodes can jointly sample the noise using the protocol of Wu et al.~\cite{wu16}.
The disadvantage of this approach is that it will add a large computational overhead.
However, it would still be significantly faster than the strawman solution of applying an end-to-end MPC protocol such as Mohassel et al.~\cite{Mohassel2020} which would also need to sample noise this way (in addition to their already high computation cost).

\subsection{Local and Shuffle DP}

Another way to alleviate the need for a semi-honest server is to switch to a local or shuffle DP model. 
However, as discussed in Section~\ref{sec:model_comparison}, this is necessarily less accurate than the IND-CDP-MPC model that attains central DP accuracy (even when $\DatasetSize_{\PartyIndex}=1$).
We note that extending \ourProtocol~to the local or shuffle DP model is not advantageous, since the data is already noised in these models, secure computations are not required.
Thus, improvements in the local and shuffle DP models are a tangential research direction and not the focus of this work.

\ifthenelse{\boolean{extended}}{
\section{Additional Utility Evaluation}\label{sec:appendix_util_eval}

We evaluate our protocol with three metrics: Silhouette Score~\cite{rousseeuw198753}, Davies-Bouldin Index (DBI)~\cite{daviesbouldin1979}, and Mean Squared Error (MSE). MSE is the average squared distance between output centroids and ground-truth centroids (from k-means++ matched via the Hungarian algorithm). Because MSE is sensitive to outliers, methods without post-processing (e.g.\ \suProtocol) can produce extreme MSE values, whereas \ourProtocol~remains close to the non-private baseline. For any single-cluster result, we set Silhouette Score to -1 and DBI to $\infty$; accordingly, infinite DBI values are omitted in their figure. Otherwise, all three metrics exhibit similar trends as reported in Section \ref{sec:eval} of the main paper.

\generateMetricFigureGroup{Silhouette}{Silhouette Score}

\generateMetricFigureGroup[false]{Davies}{Davies-Bouldin Index}

\generateMetricFigureGroup[false]{MSE}{Mean Squared Error (Centroid Alignment)}
}{}

\end{document}